\pdfoutput=1
\documentclass[manuscript,screen,nonacm]{acmart}

\usepackage{mathtools}      
\usepackage{bm}
\usepackage{amsthm}

\usepackage{graphicx}
\usepackage{xcolor}
\usepackage{booktabs}
\usepackage{array}
\usepackage{multirow}
\usepackage{tabularx}
\usepackage{colortbl}
\usepackage{wrapfig}
\usepackage{makecell}

\usepackage[english]{babel}
\usepackage{soul}
\usepackage{comment}
\usepackage{hhline}
\usepackage{xspace}
\usepackage[strings]{underscore} 

\usepackage{siunitx}        

\usepackage{todonotes}

\usepackage{hyperref}
\usepackage[ruled,vlined]{algorithm2e}
\usepackage{subcaption}

\newtheorem{problem}{Problem}
\newcolumntype{C}[1]{>{\centering\arraybackslash}p{#1}}

\newcommand{\cupdot}{\mathbin{\mathaccent\cdot\cup}}

\AtBeginDocument{%
  }

\setcopyright{acmlicensed}
\copyrightyear{2026}
\acmYear{2026}
\acmDOI{XXXXXXX.XXXXXXX}

\acmISBN{978-1-4503-XXXX-X/18/06}

\begin{document}

\title[Statistical-Symbolic Verification of Perception-Based Autonomous Systems]{Statistical-Symbolic Verification of Perception-Based Autonomous Systems using State-Dependent Conformal Prediction}



\author{Yuang Geng}
\authornote{Both authors contributed equally to this research.}
\email{yuang.geng@ufl.edu}
\orcid{0000−0002−2265−7586}
\affiliation{%
  \institution{University of Florida}
  \city{Gainesville}
  \state{Florida}
  \country{USA}
}

\author{Thomas Waite}\authornotemark[1]
\email{waitet@rpi.edu}
\orcid{0009-0009-9925-6789}
\affiliation{%
  \institution{Rensselaer Polytechnic Institute}
  \city{Troy}
  \state{New York}
  \country{USA}
}

\author{Trevor Turnquist}
\email{trevor.turnquist@ufl.edu}
\orcid{0009-0000-0900-6001}
\affiliation{%
  \institution{University of Florida}
  \city{Gainesville}
  \state{Florida}
  \country{USA}
}

\author{Radoslav Ivanov}
\authornote{Co-last authors: equally shared supervision responsibility.}
\email{ivanor@rpi.edu}
\orcid{0000-0003-4987-4836}
\affiliation{%
  \institution{Rensselaer Polytechnic Institute}
  \city{Troy}
  \state{New York}
  \country{USA}
}

\author{Ivan Ruchkin}\authornotemark[2] 
\email{iruchkin@ece.ufl.edu} 
\orcid{0000-0003-3546-414X}
\affiliation{%
  \institution{University of Florida}
  \city{Gainesville}
  \state{Florida}
  \country{USA}
}

\renewcommand{\shortauthors}{Geng et al.}

\begin{abstract}
Reachability analysis has been a prominent way to provide safety guarantees for neurally controlled autonomous systems, but its direct application to neural perception components is infeasible due to imperfect or intractable perception models. Typically, this issue has been bypassed by complementing reachability with statistical analysis of perception error, say with conformal prediction (CP). However, existing CP methods for time-series data often provide conservative bounds. The corresponding error accumulation over time has made it challenging to combine statistical bounds with symbolic reachability in a way that is provable, scalable, and minimally conservative. To reduce conservatism and improve scalability, our key insight is that perception error varies significantly with the system's dynamical state. This article proposes \textit{state-dependent conformal prediction}, which exploits that dependency in constructing tight high-confidence bounds on perception error. Based on this idea, we provide an approach to partition the state space, using a genetic algorithm, so as to optimize the tightness of conformal bounds. Finally, since using these bounds in reachability analysis leads to additional uncertainty and branching in the resulting hybrid system, we propose a branch-merging reachability algorithm that trades off uncertainty for scalability so as to enable scalable and tight verification. The evaluation of our verification methodology on two complementary case studies demonstrates reduced conservatism compared to the state of the art. 

\end{abstract}

\ccsdesc[300]{Computer systems organization~Embedded and cyber-physical systems}
\ccsdesc[500]{General and reference~Verification}
\ccsdesc[300]{Computing methodologies~Neural networks}
\ccsdesc[300]{Computing methodologies~Uncertainty quantification}

\keywords{neural network verification,  pixel-to-action control, conformal prediction}

\maketitle

\section{Introduction}

Neural perception and control components are now widely deployed in autonomous systems, from driverless cars to air taxis and delivery robots~\cite{waymo, volocity, serveRobotics2025}. 
These systems use high-dimensional sensors (cameras, LiDAR) to perceive their environment and make safety-critical decisions, and while neural perception and control components have demonstrated impressive capabilities in executing these complex tasks, their unpredictable behavior can lead to serious accidents~\cite{waymoRecall2025, cruiseAccident2024}. Therefore, providing formal safety guarantees for such systems is essential. However, providing safety guarantees for such systems is difficult due to the complex data and perception components required to operate in stochastic, high-dimensional environments. 

Formal methods like reachability analysis offer a principled approach to verify safety~\cite{dutta2019reachability}. However, they face a fundamental challenge: accurately modeling high-dimensional sensor data and perception pipelines. While, explicit perception models based on physics (e.g., pinhole camera models, LiDAR ray tracing) are tractable. They cannot capture the stochastic complexity of real-world sensing (e.g., reflected LiDAR rays) ~\cite{habeeb2024interval,santa2023certified, ivanov20a}. Learned generative models can be used to capture this complexity~\cite{katz22}, but verifying high-dimensional, pixel-level models in a closed-loop setting remains difficult to scale~\cite{everett2021neural}. Moreover, the validity of generative models as surrogates for the real environment remains in question.
On the other hand, data-driven methods like conformal prediction (CP) can provide probabilistic guarantees that capture the uncertainty of the environment, yet these approaches often fail to exploit knowledge of the system (e.g., dynamics) and result in overly conservative safety guarantees~\cite{pmlr-v242-chee24a}.

Emerging approaches for providing safety guarantees on neural perception and control systems use a \textit{combination} of data-driven and symbolic methods~\cite{lin2024verification,muthali2023multi,geng2024bridging, waite2025stateconform}. Such approaches use data-driven methods to get statistical uncertainty bounds on difficult-to-model quantities (e.g., the neural perception error) and then use those bounds in symbolic reachability analysis with analytical descriptions of known components (e.g., dynamics) to produce high-confidence reachable sets. Effectively, this two-step approach can abstract away the high-dimensional perception models while preserving knowledge of the system. 

However, when the combination of symbolic and data-driven methods is applied naively, it can result in both \textit{unnecessary conservatism} and \textit{poor scalability}, especially in the challenging case of time-series data. If the data-driven bounds are too loose, then the conservatism compounds over time during reachable set computation, particularly as over-approximation error accumulates over long time horizons. A variety of methods have been proposed to reduce conservatism in conformal bound size \cite{romano2019conformalized, sharma2024pac,kiyani2024length,tumu2024multi}, with a particularly promising approach for reducing conservatism in conformal bounds that finds an optimal reweighting of conformal scores over time \cite{cleaveland2024conformal}. This time-based optimization produces much tighter bounds as compared to point-wise bounds, which must rely on a union bound per timestep \cite{lindemann2023safe}. In practice, the optimized bounds can exploit the heteroskedasticity of the prediction errors over time, reducing the impact of particularly high-error examples on the remainder of the trajectory. While reducing bound sizes using a method like this should help reduce conservatism in the downstream reachability task, existing approaches do not explicitly optimize bounds with reachability in mind. 

Even if the bounds are reasonably tight for statistical purposes, verification tools still may not be able to produce a valid reachable due to well-known scalability challenges of verifying systems under persistent uncertainty (e.g.,~state-space explosion)~\cite{alur2011formal, taleb2023uncertainty, gao2024robust, frehse2011spaceex}. Once again, there has been some work to improve verification scalability under uncertainty \cite{maiga2015comprehensive, ramdani2009hybrid, ghosh2019robust}, but it remains a challenging problem. These two-sided issues suggest the need for a \textit{targeted interface} between the statistical and symbolic steps, to simultaneously address conservatism in data-driven bounds and increase the scalability of reachable set computation under uncertainty.

Though optimizing conformal bounds to exploit perception error heteroskedasticity over time \cite{cleaveland2024conformal} seems like an intuitive place to start interfacing statistical and symbolic methods,  perception error is not always correlated with \textit{time}. Instead, heteroskedasticity in error is often better correlated with the \textit{state}. For example, camera images may be blurred at high speed, or LiDAR rays may be reflected in particular orientations of the sensor with its environment. Therefore, perception error is also often heteroskedastic with respect to the \emph{state}.

In our preliminary work \cite{waite2025stateconform}, we applied this insight and introduced \textbf{state-dependent conformal bounds} for neural perception error as an approach to reduce conservatism in safety analysis of systems with known dynamics and neural perception in the loop. 
Our approach exploits error heteroskedasticity over the state space by dividing it into regions. We then propose an optimization problem to select optimal regions such that regional perception errors contribute minimally to over-approximation error in symbolic reachability computation. Although promising as a first step, the original approach suffers from scalability limitations due to overly conservative bounds and rapid accumulation of approximation error during verification.

In this work, we explore two complementary techniques to improve the utility and scalability of the statistical-symbolic approach. First, we improve the conformal bound sizes and apply them to higher-dimensional systems. To do this, we refine our optimization problem to include multi-dimensional state regions as well as dynamic per-region confidences. Multi-dimensional state partitioning permits the use of our method on more realistic, higher-dimensional systems, and dynamic per-region confidences improve our ability to exploit the heteroskedasticity (as compared to the fixed confidences in our preliminary work) while still maintaining the trajectory-wide, user-defined confidence level. 

Second, we provide verifiability improvements to the symbolic reachability analysis of hybrid systems under uncertainty. 
In hybrid systems, if the reachable sets are large or subject to persistent uncertainty, they can enter multiple modes with distinct dynamics (e.g., ``cruise control'' mode or ``avoid obstacle'' mode) at the same time. When this happens, each resulting branch must be verified separately, leading to prohibitively long run times. This challenge is exacerbated in our setting, where we must inflate our reachable sets to account for the statistical uncertainty as well as introduce additional modes for each region. To address this challenge, we introduce a \textit{branch-merging algorithm} that consolidates similar branches (with overlapping reachable sets) during verification to permit verification of previously unverifiable highly-branching systems.

We evaluate our approach on two case studies. First, a Mountain Car \cite{geng2024bridging} with an image-based neural perception model estimating a low-dimensional state. We find optimal regions in a single dimension of the state space and compute reachable sets under a neural network controller. The second case study is a 10th-scale autonomous racing car \cite{roboracer2025} with LiDAR-based neural perception operating in a higher-dimensional state space. We compute optimal, two-dimensional regions and compute reachable sets under a PID controller. This case study demonstrates the scalability of our methods to more realistic multi-dimensional systems with neural perception. 
In both case studies, our methods achieve significantly smaller reachable set sizes compared to the time-based approach \cite{cleaveland2024conformal}.
Our combined contributions are fourfold:

\begin{enumerate}
    \item \textbf{State-based safety guarantees for neural perception systems.} We introduce a methodology to exploit heteroskedastic perception error to obtain user-specified safety guarantees, supported by a conformal-prediction method that yields tighter, region-based perception error bounds.
    
    \item \textbf{Region partitioning for tight statistical bounds.} We introduce an algorithm to partition multi-dimensional state spaces for conformal prediction with per-region confidence allocations. 

    \item \textbf{Improved reachability under uncertainty.} We provide a branch-merging algorithm that trades off uncertainty for scalability and thus enables more scalable reachability analysis of hybrid systems.
    
    \item \textbf{Closed-loop evaluations on low and high-dimensional systems.} We demonstrate the benefits of our approach on both the Mountain Car benchmark with image-based perception and a LiDAR-based autonomous racing system with a multi-dimensional state space.
\end{enumerate}

In Section~\ref{sec:bg_and_relwork}, we provide essential background and related work. In Section~\ref{sec:prob} we formalize our problem statements. Section~\ref{sec:approach} overviews our approach, and Sections~\ref{sec:state_regions}~and~\ref{sec:reachability} detail our methods for optimizing conformal bounds and improving hybrid system verifiability, respectively. In Sections~\ref{sec:evaluation}~and~\ref{sec:conclusion}, we present the evaluation of our methods in two case studies and discuss results and next steps.

\section{Background and Related Work}\label{sec:bg_and_relwork}
Before introducing the problem setup, we present the necessary background and related work to establish our context.
We begin by defining our dynamical system with perception components as the conceptual basis for the entire analysis, as motivated by prior work on perception-based robust control~\cite{dean20}:
\label{sec:problem}
\begin{align}
\label{eq:system_model}
\begin{split}
    x_{k+1} = f(x_k, u_k);\quad 
    z_k = p(x_k);\quad 
    y_k = nn(z_k) := g(x_k) + v_k;\quad 
    u_k = h(y_k),
\end{split}
\end{align}
where $x_k \in \mathcal{X}$ are the system states (e.g.,~position, velocity); $z_k$ are the measurements (e.g., images) generated from an unknown perception map $p$ (e.g., camera); $y_k$ are the outputs of a perception neural component $nn$ trained to extract task-relevant low-dimensional state $g(x_k)$ (e.g., estimated positions); 
$v_k$ is the unknown perception noise introduced by the neural component $nn$; 
$f$ is the known plant dynamics model, and $h$ is a known and fixed controller.

The above system definition simplifies both control design and verification.
Due to decomposing perception as $y_k = g(x_k) + v_k$, the perception component estimates a low-dimensional, task-relevant function of the state (e.g., a car’s lateral position in the lane), while all uncertainty is captured entirely by the perception error term $v_k$.
This simplifies controller design, since the controller can be developed for $g(x_k)$ directly and made robust to bounded perception error $v_k$~\cite{dean20}.
It also simplifies verification, since reasoning about safety only requires bounding $v_k$, rather than modeling the full perception map. 
Together, these properties provide a solid foundation for high-confidence control and safety verification while abstracting away the complexity of perception.
Next, we present our background and related work.

\paragraph{Reachability Analysis.} 
Reachability analysis computes the set of all possible states that the a system can reach over time from an initial set. It provides safety guarantees by checking whether these reachable sets intersect with the unsafe region~\cite{huang2022polar}.
Reachability methods are typically divided into two categories: open-loop approaches, which focus on verifying input-output properties of neural networks~\cite{wang2021beta, dutta18, katz17, tran20}, and closed-loop approaches, which interleave neural networks with symbolic dynamics to compute reachable sets over multiple time steps~\cite{ivanov19, dutta2019reachability, huang2019reachnn, fan2020reachnn, wang2023polar}.

Our work focuses on closed-loop reachability analysis to provide safety guarantees of the full system over a specified time horizon. 
Classical reachability analysis requires specifying uncertainty on noise terms, typically as a worst-case bound that holds at all time steps.
For the system in~\eqref{eq:system_model}, given the controller $h$, dynamics $f$, an initial set $\mathcal{X}_0$, and a noise bound $b \ge \|v_k\| $, we compute sets $\mathcal{X}_1, \dots, \mathcal{X}_T$ that are guaranteed to contain all possible states $x_k$ at each time step $k = 1..T$ and check that no unsafe states are included.
These sets are typically conservatively approximated with tractable shapes such as ellipsoids \cite{althoff15} or Taylor models (TMs) \cite{wang2023polar, chen12, ivanov2021b}. TMs are the reachability representation used in this work. Introduced by~\cite{makino03}, TMs are particularly useful as they form valid over approximations for arbitrary differentiable functions, can be propagated through non-linear system dynamics, and have well-defined operations (addition, subtraction, multiplication, composition, etc). Formally, they are defined as follows: 

\begin{definition}[Taylor Model]
\label{def:TMs}
Given a $d$-times differentiable function $f: \mathbb{R}^n \rightarrow \mathbb{R}$ on a domain $D \subset \mathbb{R}^n$, an order $d$ Taylor Model is a tuple $T=(p, I)$ such that $p$ is a polynomial approximation of $f$ of degree $d$ and $I$ is an interval error bound of $p$ over $D$, with the following enclosure property: $\forall x\in D: \; f(x) \in \{p(x) + e \mid e \in I\}$.

\end{definition}

As mentioned above, standard reachability analysis requires bounds on noise $v_k$, ideally over the worst case. Obtaining such worst-case bounds, however, may require strong assumptions.
For instance, there exists work to identify failures of vision-based controllers with Hamilton-Jacobi reachability, but it relies on a manually chosen noise bound, which makes it impossible to get formal safety guarantees~\cite{chakraborty2023discovering}.
This limitation motivates a probabilistic perspective: instead of guaranteeing coverage for all possible states, we seek to compute reachable sets that contain the system’s state \emph{with high confidence} over random initial conditions and noise realizations. 
Such probabilistic reachable sets can support high-confidence pre-deployment guarantees, real-time safety monitoring, and motion planning in interactive environments.

\paragraph{Scalar Conformal Prediction.}
One approach to obtaining high-confidence noise bounds is via CP~\cite{vovk2005algorithmic,shafer2008tutorial}.
CP is an increasingly popular method for obtaining data-driven uncertainty bounds. 
As CP has expanded to a variety of applications, its role in autonomous system safety has gained particular attention, including safe motion planning \cite{lindemann2023safe}, safe controller design \cite{yang2023safe}, online safety monitoring \cite{zhang2024bayesian,zhao2024robust}, and integration into safety decision-making frameworks~\cite{lekeufack2024conformal}. 
We refer the reader to recent tutorials~\cite{lindemann2024formal,angelopoulos2023conformal} for a broader overview of the CP area.

In scalar CP~\cite{angelopoulos2023conformal}, we are given a calibration dataset $D = \{z_1, \dots, z_N\}$, where the $z_i$ are realizations of exchangeable random variables $Z_1, \dots, Z_N$.
Exchangeability means that the joint distribution is invariant under any reordering, ${\mathbb{P}[Z_{q(1)} \le \dots \le Z_{q(N)}] = \mathbb{P}[Z_{r(1)} \le \dots \le Z_{r(N)}]}$ 
for any two re-ordering functions $q$ and $r$. 
Now consider a new random variable $Z_{test}$ that is also exchangeable with $Z_1, \dots, Z_N$. 
Under this assumption, the rank of $Z_{test}$ among $N+1$ samples is uniformly distributed.
Therefore, if we sort $z_i$ in increasing order, resulting in ordered $z_{(1)}, z_{(2)}, \ldots , z_{(N)}$, we obtain
\begin{align}
{\mathbb{P}[Z_{test} \le \text{Quantile}(D, 1-\alpha)] \ge 1 - \alpha}, \quad \text{ where } ~ \text{Quantile}(D, 1-\alpha) :=  z_{( \lceil (N+1)(1-\alpha) \rceil )}. 
\end{align}
In other words, the (normalized) $1-\alpha$ quantile, which is the order statistic $z_{( \lceil (N+1)(1-\alpha) \rceil )}$ denoted by $\text{Quantile}(D, 1-\alpha)$, serves as a high-confidence upper bound for any new exchangeable sample $Z_{test}$.

\paragraph{Time-Series Conformal Prediction.} 
Our closed-loop verification of system~\eqref{eq:system_model} extends over a specified time horizon and therefore requires time-series probabilistic noise bounds. However, standard scalar CP is not directly applicable to time series because observations at different time steps are temporally dependent rather than independent~\cite{2023CPtime, cleaveland2024conformal}.
A common workaround is to apply CP independently at each time step and then use a union bound to combine the results, ensuring overall coverage across the entire time horizon~\cite{oliveira2024split}.
This approach requires setting the per-step confidence level to $1- \alpha/T$ in order to achieve an overall guarantee $\alpha$~\cite{angelopoulos2023conformal}.
However, as $T$ increases, this conservatism compounds, causing the resulting bounds to become too conservative to be practical for planning or safety-critical applications.

Therefore, recent time-series CP methods focus on modeling entire trajectories rather than individual time steps, in order to account for temporal dependence and move beyond the exchangeability assumption. Given a dataset $D = \{z_{1,0:T}, \dots, z_{N,0:T}\}$ of $N$ trajectories, each defined over a fixed time horizon $T$, the goal is to design a trajectory-wide, possibly time-based, bound $b_k$ such that ${\mathbb{P}[\forall k = 0..T: Z_{test,k} \le b_k]} \ge 1-\alpha$,
%
meaning that the entire test trajectory $Z_{test}$ remains within the time-series bound $b_k$ with high confidence. 
To improve upon the naive union-bound-based solution described above, researchers~\cite{cleaveland2024conformal,angelopoulos2023conformal_pid} have proposed to re-weigh the bounds at each time step to tighten up the bounds.
For instance, Cleaveland et al.~\cite{cleaveland2024conformal} introduce a parameterized non-conformity score optimized through linear complementarity programming to construct valid prediction regions over entire trajectories.
In contrast, our work focuses on the state space partitioning that adaptively allocates the conformal bounds, thereby achieving tighter and less conservative reachability guarantees compared to other approaches.

\paragraph{Probabilistic Verification.}
Our work extends \textit{non-probabilistic} reachability analysis mentioned above by explicitly accounting for uncertainty in perceptual disturbances with \textit{statistical} methods. 
Alternative approaches, such as probabilistic model checking with PRISM~\cite{PRISM}, estimate the probability that a system trajectory reaches an unsafe state within a finite horizon. 
These methods provide rigorous guarantees but require explicit low-dimensional probabilistic models, making them difficult to apply to our high-dimensional, data-driven system in~\eqref{eq:system_model}~\cite{mitra_formal_2025}. This challenge has led to parallel ideas of abstracting perceptual disturbance with confidence intervals~\cite{pasareanu_closed-loop_2023,calinescu_controller_2024} and constricting Interval Markov Decision Process (IMDP) models for safety verification~\cite{cleaveland_conservative_2025,peper_towards_2025}. Our research elaborates these ideas to reduce conservatism in non-probabilistic reachability, which tightly captures the behavior of low-dimensional neural controllers.   

\paragraph{Model-Free Reachability.}
Recently, a variety of data-driven reachability techniques have emerged without the need for first-principles models. For instance, several works~\cite{Data-Driven_RA_USC, zhou2024conformalized} estimate reachable sets directly from sampled trajectories or empirical distributions of disturbances. 
These approaches bypass the need for models by deriving statistical bounds from observed data, allowing reachability analysis to be applied to complex, black-box, or learned dynamics~\cite{CLEAVELAND2022103782, GeneralRA}.
Compared to these methods, we derive state-dependent, data-driven bounds and integrate them into model-based reachability analysis. Our approach benefits from the knowledge of the dynamical model, producing tighter reachable sets than existing state-of-the-art techniques.

\section{Problem Formulation}\label{sec:prob}

Given the system in~\eqref{eq:system_model}, the unknown perception component $p$ (e.g., a camera) plays a major role in the control process.
However, directly reasoning about its complex, high-dimensional structure is challenging for safety verification.
To make the verification more tractable, we abstract the effect of $p$ on the control loop rather than modeling $p$ explicitly.
Specifically, we represent the output of $p$ through the desired lower-dimensional state estimates $g(x_k)$ and an additive noise term $v_k$.
This abstraction preserves the essential influence of $p$ on the system while enabling scalable safety verification.
The resulting abstracted system is defined below: 
\begin{align}
\label{eq:abst_system_model}
\begin{split}
    x_{k+1} = f(x_k, u_k); \quad y_k = g(x_k) + v_k; \quad  u_k = h(y_k),
\end{split}
\end{align}
Given the abstracted system in~\eqref{eq:abst_system_model}, our \textbf{main goal} is to provide high-confidence safety guarantees through statistical-symbolic reachability analysis. 
Specifically, we aim to construct a sequence of reachable sets 
$\mathcal{X}_1, \dots, \mathcal{X}_T$ that contain the system states at all time steps with confidence at least $1 - \alpha$. 

To support this analysis,
we assume access to a dataset of $N$ trajectories, $D = \{(x_{1,0:T}, y_{1,0:T}), \dots, (x_{N,0:T}, y_{N,0:T})\}$, sampled from the system in~\eqref{eq:abst_system_model}.
Each trajectory $i$ starts from an initial state $x_{i,0}$ drawn independently and identically distributed (IID) according to an unknown distribution $\mathcal{D}_0$ over initial set $\mathcal{X}_0$, and evolves for $T$ steps\footnote{While our method does not require different trajectories to have the same time horizon, we fix the horizon for simplicity of presentation.} under a perception noise sequence $\mathbf{v}_i = [v_{i,0},\dots,v_{i,T}]$ drawn from an unknown joint distribution $\mathcal{V}$. This process produces the dataset
where \( x_{i,0:T} \) denotes the true system states along the \( i \)-th trajectory from time step \( 0 \) to \( T \), and 
\( y_{i,0:T} \) denotes the corresponding neural component outputs.

\begin{problem}[High-Confidence Reachability]
\label{prob:reachability}
Given the abstracted system in~\eqref{eq:abst_system_model}, 
a confidence level $\alpha$, and a trajectory dataset $D$, 
the goal is to \emph{construct a sequence of reachable sets} 
$\mathcal{X}_1, \dots, \mathcal{X}_T$ such that
\[
\mathbb{P}_{x_0 \sim \mathcal{D}_0,\, \mathbf{v} \sim \mathcal{V}}
\Big[\forall k \in \{0,\dots,T\}: x_k \in \mathcal{X}_k \Big] \ge 1 - \alpha.
\]

\end{problem}

As argued in the previous section, solving Problem~\ref{prob:reachability} requires statistical bounds for perception noise $\mathbf{v}$, acquired for instance from CP~\cite{lindemann2024formal,angelopoulos2023conformal}.
However, the current scalar or time-based CP would yield overly conservative bounds~\cite{cleaveland2024conformal}, leading to high and unnecessary inflation in reachable sets.
To address this, we aim to exploit heteroskedasticity in the perception error over the \textbf{state space}. This idea leads to the task of partitioning the state space into low- and high-error regions, to minimize the accumulation of approximation error during reachability analysis.

\begin{problem}[State Partitioning for Conformal Noise Bounding]
\label{prob:state_dependent_cp}
Given the same setting as in Problem~\ref{prob:reachability}, the goal is to partition the state space $\mathcal{X}$ into $M$ disjoint regions,
${\mathcal{X} = \mathcal{S}_1 \cupdot  \cdots \cupdot  \mathcal{S}_M},$
and to assign a confidence level $\alpha_i$ to each region $\mathcal{S}_i$ such that $\sum_{i=1}^M \alpha_i = \alpha.$
Then, for each region, the goal is \emph{to compute a region-based noise bound} $\eta_i$ such that
\[
\mathbb{P}_{x_0 \sim \mathcal{D}_0,\, \mathbf{v} \sim \mathcal{V}}
\Big[\forall k \in \{0,\dots,T\}:\;
x_k \in \mathcal{S}_i \Rightarrow \|v_k\| \le \eta_i \Big] 
\ge 1 - \alpha_i,
\]
thus ensuring the overall reachability confidence of $1-\alpha$  (shown in Section~\ref{sec:approach} with a union bound).
The partitioning should minimize a loss function $\mathcal{L}\big(D, \{\mathcal{S}_i\}_{i=1}^M,\{\alpha_i\}_{i=1}^M\big)$ that correlates with the tightness of the reachable sets $\mathcal{X}_k$ in Problem~\ref{prob:reachability}.
\end{problem}
\noindent
\textit{Remark}:
Problem~\ref{prob:state_dependent_cp} has two parts: 
1) identifying a suitable loss function $\mathcal{L}$ for optimizing the partition; 
2) solving the resulting optimization to obtain state-based bounds. 
These bounds are then integrated into the reachability analysis, resulting in a solution of Problem~\ref{prob:reachability} with \emph{tighter}, \emph{less conservative} guarantees.

\section{Approach Overview}
\label{sec:approach}


\begin{figure}[t]
\captionsetup{justification=centering}
    \centering
    \includegraphics[width=\textwidth]{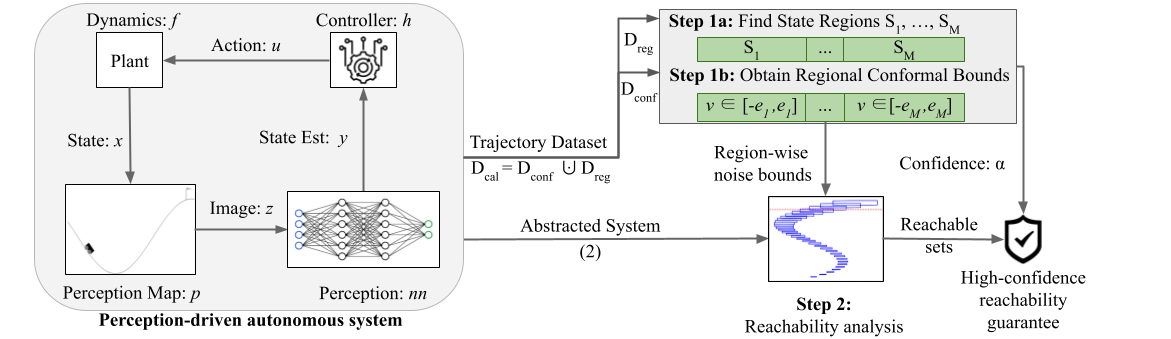}
    \caption{High-confidence reachability approach overview. 
    \textbf{Step 1 (Partition region):}
    Using perception errors along trajectories, we partition the state space to minimize the trajectory-based perception bounds and compute region-specific conformal bounds.
    \textbf{Step 2 (Reachability analysis):} We inflate the reachability analysis with these bounds to obtain high-confidence reachable sets.}
    \label{fig:approach-overview}
\end{figure}

We propose a \textbf{two-step approach} to the main Problem~\ref{prob:reachability}, as illustrated in Figure~\ref{fig:approach-overview}.
The key idea is first to estimate region-based perception noise bounds (Problem~\ref{prob:state_dependent_cp}) and then incorporate these bounds into reachability analysis to solve Problem~\ref{prob:reachability}. 
\textbf{Step 1} optimizes the state-space partition with a designed loss function 
and applies CP to provide a high-confidence noise bound for each region.
\textbf{Step 2} then integrates these region-based bounds into the worst-case reachability analysis, producing high-confidence reachable sets. 
By decoupling uncertainty modeling from reachability computation, this two-step approach achieves statistical validity while maintaining the tight reachable set approximations.

\vspace{5px}
\noindent
\textbf{Step 1 (Region-based CP):}
Since perception error bounds vary by regions, the first step is to partition the whole state space into $M$ subregions $\mathcal{X} = \mathcal{S}_1 \cupdot \dots \cupdot \mathcal{S}_M$, and assign each region $\mathcal{S}_i$ its own confidence level $1-\alpha_i$. 
The per-region confidence is chosen such that $\sum_{i=1}^{M}\alpha_i = \alpha$, where $\alpha$ is the desired global confidence level selected by the user.
This allows us to construct \textit{state-based perception noise bounds}.
Specifically, for each $\mathcal{S}_i$, we calculate a bound $\eta_i$ such that, with probability at least $1-\alpha_i$,
the perception error remains in the bounds whenever the state is in $\mathcal{S}_i$ over the horizon:
\begin{align}
\label{eq:one_region_confidence}
\mathbb{P}_{x_0\sim \mathcal{D}_0,\; \mathbf{v} \sim \mathcal{V}}
\Big[
\forall k\in\{0,\dots,T\} : x_k\in\mathcal{S}_i \implies \ 
\|v_k\| \le \eta_i
\Big] \;\ge\; 1-\alpha_i .
\end{align}

Next, by applying the \textit{union bound} (Boole’s inequality) across all subregions, we can merge the region-based guarantees into a single trajectory-wide global guarantee. 
This ensures a \textbf{global confidence level} of $1 - \alpha$ while still maintaining tight bounds within each region. To conveniently represent our bounds across state regions, we define a global \textbf{piecewise noise-bound function} $ H: \mathcal{X} \to \mathbb{R}_{\ge 0}:$
\begin{align}
    \label{eq:global_bound}
    H(x) := \eta_i, \text{ if } x \in \mathcal{S}_i.
\end{align}

Unlike a \textit{constant} global worst-case noise bound, this \textit{piecewise-constant} global bound $H(x)$ reflects the local levels of perception noise, producing tighter overall reachable sets. This reasoning is formalized in the following proposition.

\begin{proposition}[Global Partitioning Guarantee]
\label{prop:sd_conformal_prediction}
Consider the abstracted system in~\eqref{eq:abst_system_model} and a state-space partitioning $\mathcal{X} = \mathcal{S}_1 \cupdot \cdots \cupdot \mathcal{S}_M$, with $\eta_i$ as the region-based noise bound with confidence $\alpha_i$ for each $\mathcal{S}_i$, per~\eqref{eq:one_region_confidence}, and $H(x)$ as the global noise bound, per~\eqref{eq:global_bound}.  
If the regional confidences satisfy $\sum_{i=1}^M \alpha_i = \alpha$, the following \emph{trajectory-wide noise bound} holds:
\[
\mathbb{P}_{x_0\sim \mathcal{D}_0,\, \mathbf{v}\sim \mathcal{V}}
\Big[ \forall k \in \{0,\dots,T\}:\;
\|v_k\| \le H(x_k) \Big] \ge 1- \alpha.
\]
\end{proposition}

\begin{proof}
For each region $\mathcal{S}_i$, define the violation event
$
E_i := \{\exists\, k \in \{0,\dots,T\}:\; x_k \in \mathcal{S}_i \ \wedge\ \|v_k\| > \eta_i \}.
$
From the proposition statement, we know that $\mathbb{P}(E_i) \le \alpha_i$ for all $i=1,\dots,M$. Let
$
A := \cup_{i=1}^M E_i
= \{\exists\, k:\; \|v_k\| > H(x_k)\}.
$
Then, by the union bound,
$
\mathbb{P}(A) \le \sum_{i=1}^M \mathbb{P}(E_i) \le \sum_{i=1}^M \alpha_i = \alpha.
$
Taking complements yields the desired lower bound on the satisfaction chance:
$
\mathbb{P}[\forall\, k:\; \|v_k\| \le H(x_k)]
= 1 - \mathbb{P}(A) \;\ge\; 1-\alpha.$
\end{proof}

\noindent
\textbf{Step 2 (High-Confidence Reachability):} 
Next, we perform worst-case (deterministic) reachability analysis of the abstracted system in~\eqref{eq:abst_system_model} 
from initial set $\mathcal{X}_0$.
The region-based perception error bounds $H(x)$ from Step~1 are incorporated as noise bounds in the reachability, 
so that the resulting reachable sets are valid with high probability. 
Therefore, Step 1 ensures that the statistical guarantee on perception errors leads to valid worst-case reachable sets in Step 2, thereby solving Problem~\ref{prob:reachability}. 
The following Theorem~\ref{thm:reachability} provides the overall safety guarantee.

\begin{theorem}[High-Confidence Reachability Guarantee]
\label{thm:reachability}
\looseness=-1
Consider the abstracted system in~\eqref{eq:abst_system_model},
with initial states $x_0$ sampled from a distribution $\mathcal{D}_0$ 
supported on an initial set $\mathcal{X}_0$. Suppose there exists a global confidence bound $H(x)$ s.t.:

\[
\mathbb{P}_{x_0 \sim \mathcal{D}_0,\, \mathbf{v} \sim \mathcal{V}}
\Big[\forall k \in \{0,\dots,T\}: \|v_k\| \leq H(x_k) \Big] \geq 1-\alpha.
\]
Let $\mathcal{X}_1, \dots, \mathcal{X}_T$ be the worst-case reachable sets 
of the system in~\eqref{eq:abst_system_model},
computed from $\mathcal{X}_0$ under disturbance bounds $\|v_k\| \le H(x_k)$ for all $k=1..T$.
Then, with probability at least $1 - \alpha$ over the joint distribution 
$(x_0, \mathbf{v}) \sim (\mathcal{D}_0, \mathcal{V})$, the true trajectory $x_k$ remains within the corresponding reachable sets: $\mathbb{P}_{x_0 \sim \mathcal{D}_0,\, \mathbf{v} \sim \mathcal{V}}
[\forall k \in \{0,\dots,T\}: x_k \in \mathcal{X}_k ] \ge 1 - \alpha.$
\end{theorem}
%
\begin{proof}
Consider the event $A = \{\exists k = 0..T: x_k \notin \mathcal{X}_k\}$. Since the $\mathcal{X}_i$ are worst-case reachable sets, it must be the case that $\|v_l\| > H(x_l)$ for some $l \le k$, i.e., 
$A = \{\exists k = 0..T: (x_k \notin \mathcal{X}_k \wedge \exists l \le k: \|v_l\| > H(x_l))\}.$
However, we know that the noise bounds hold with probability $1-\alpha$ over the entire trajectory, so $\mathbb{P}_{x_0\sim \mathcal{D}_0, \mathbf{v} \sim \mathcal{V}}[A] \le \alpha$, and the result follows.
\end{proof}

\section{Region-based Perception Error Bounds}
\label{sec:state_regions} 

This section presents an optimization-based approach to find a partitioning $\mathcal{X} = \mathcal{S}_1 \cupdot \dots \cupdot \mathcal{S}_M$ that 
(i)~satisfies the probability guarantee in Proposition~\ref{prop:sd_conformal_prediction} and
(ii)~reduces the approximation error incurred by reachability analysis. 

\subsection{Regional Conformal Bounds for Trajectory-Wide Guarantees}
We first introduce the trajectory dataset collected from the system~\eqref{eq:system_model}, 
following the typical setup where each trajectory captures the true states and 
perception outputs under the closed-loop policy: 
$D =\{(x_{1,0:T}, y_{1,0:T}), \dots, (x_{N,0:T}, y_{N,0:T})\}.$ 
Each trajectory is generated under controller $h$ to ensure the on-policy IID setting. Initial states $x_0$ are sampled from an unknown distribution $\mathcal{D}_0$ supported on $\mathcal{X}_0$, and the sequence of perception noise $\mathbf{v} = [v_0, \dots, v_T]$ is drawn from a joint distribution $\mathcal{V}$, where the noise values within a trajectory are possibly correlated. 
Since we expect perception errors to correlate more strongly with state regions than with time, we derive \textit{region-based dataset} by extracting all state–output pairs from 
$D$ whose states fall within a given region $\mathcal{S}_i$:
\begin{align}
D_{\mathcal{S}_i} = \big\{(x_{j,t}, y_{j,t}) \mid x_{j,t} \in \mathcal{S}_i \big\}.    
\end{align}
Here, $j$ is the index of the trajectory and $t$ is the index of the time step. 
This aggregation collects all state-measurement pairs that fall within the subregion $\mathcal{S}_i$, preserving their time indices $t$ for later use in the loss function.

For each subregion $\mathcal{S}_i$, based on the dataset $D_{\mathcal{S}_i}$, we next calculate the nonconformity scores for conformal prediction.
For every trajectory $j$ that visits this region, we consider the \textit{sub-trajectory} (possibly non-contiguous) lying within $\mathcal{S}_i$.
The nonconformity score is defined as the maximum perception error observed along its sub-trajectory:
\begin{align}
\delta^j_{\mathcal{S}_i} = \max_{(x_{j,t},y_{j,t}) \in D_{\mathcal{S}_i}} \| g(x_{j,t}) - y_{j,t}\|.
\end{align}
Repeating this for all relevant trajectories $j=1..N$ results in the set of nonconformity scores $\Delta_i = \{\delta^1_{\mathcal{S}_i}, \dots, \delta^{N+1}_{\mathcal{S}_i}\}$, where the additional $\delta_{\mathcal{S}_i}^{N+1}$ is set to $+\infty$ to ensure valid conformal coverage.
Finally, we apply scalar conformal prediction to $\Delta_i$ to obtain the \textit{region-specific perception error bound} $\eta_i$:
\begin{align}
\eta_i = \text{Quantile}(\Delta_i, 1 - \alpha_i).
\end{align}

\subsection{Optimization of Region Partitioning and Confidence Allocation}
Next, we partition the state space to minimize the overall conformal bound across regions.
Since these bounds directly inflate the reachable sets, careful partitioning would lead to tighter reachable set approximations in the subsequent reachability analysis.
To correlate with reachability tightness, our objective reflects:
\begin{enumerate}
    \item \textit{Region visitation frequency}: Regions visited more often contribute more to the loss since they dominate the reachability analysis time and overapproximation error.
    \item \textit{Early timesteps}: The early steps matter more because approximation error snowballs over time.
    \item \textit{Region-specific confidence levels}:
    We allocate lower confidence to high-error regions to tighten their bounds and higher confidence to low-error regions, resulting in little bounds inflation. 
    By redistributing confidence while keeping $\sum_i \alpha_i = \alpha$, we achieve tighter overall bounds without losing coverage.
    
\end{enumerate}

Putting these heuristics together, we define the following loss to optimize the partition and its per-region bounds:
\begin{equation}\label{eq:loss}
    \mathcal{L}\big(D, \{\mathcal{S}_i\}_{i=1}^M,\{\alpha_i\}_{i=1}^M\big) = \sum_{i=1}^M\sum_{x_{j,t} \in D_{\mathcal{S}_i}} \gamma^t w_i \eta_i,
\end{equation}
where the weights are $w_i = |D_{\mathcal{S}_i}|$.
This objective combines three quantities that tighten the inflation of reachability:
\begin{itemize}
    \item \textbf{Region visitation weight ($w_i$).}  Each region $\mathcal{S}_i$ is assigned a weight $w_i = |D_{\mathcal{S}_i}|$. This weight counts how many data points fall in this region. Regions with more points contribute more to the objective. This prioritizes areas that our system visits frequently or spends more time in.
    
    \item \textbf{Time decay weight ($\gamma^t$).} Every sample at time $t$ is scaled by a factor $\gamma^t$. The factor $\gamma^t$ decreases exponentially with $t$ (e.g., $\gamma^t=0.9^t$), so early timesteps receive higher importance. This reduces the approximation errors accumulated in reachability analysis over time.

    \item \textbf{Per-region confidence ($\alpha_i$).} Each region $\mathcal{S}_i$ is assigned its own confidence level $1-\alpha_i$. 
    The conformal bound $\eta_i$ is then calculated based on this confidence. 
    This design allows the optimizer to assign tighter bounds to low-error regions and looser bounds to high-error ones.
    Such adaptive allocation helps better distinguish heterogeneous uncertainty across the state space and yields overall tighter conformal bounds under the global confidence constraint.

\end{itemize}

The parameters of the optimization include both the \emph{partitions} $\{\mathcal{S}_i\}_{i=1}^M$ and the \emph{per-region confidence levels} $\{\alpha_i\}_{i=1}^M$. 
By jointly tuning these two sets of variables, the loss $\mathcal{L}$ is minimized while still respecting a global confidence constraint (i.e., $\sum_{i=1}^M \alpha_i \leq \alpha$). 
We are now ready to state the region-based optimization problem considered in this paper.

\begin{definition}[Reachability-Informed Region Optimization]
Given a calibration dataset of trajectories $D$ and a total confidence bound $\alpha$, \emph{the reachability-informed region optimization problem} is to select $M$ regions that: 
\begin{align}\label{eq:optimziation_cp_region}
    \begin{split}
    \min_{\{\mathcal{S}_i\}_{i=1}^M,\ \{\alpha_i\}_{i=1}^M}\quad 
    & \mathcal{L}\big(D, \{\mathcal{S}_i\}_{i=1}^M,\{\alpha_i\}_{i=1}^M\big) \\
    \text{s.t.} &\ \mathcal{X} = \mathcal{S}_1 \cupdot \dots \cupdot \mathcal{S}_M, \quad \alpha_i \in (0, \alpha),\quad \sum_{i=1}^M \alpha_i \le \alpha.\\
    &\ \Delta_i = \{\delta^1_{\mathcal{S}_i}, \dots, \delta^{N+1}_{\mathcal{S}_i}\}, i = 1, \dots, M, \text{ and }\\
    &\ \ \eta_i = \text{Quantile}\left(\Delta_i, 1-\alpha_i \right), i = 1, \dots, M.\\
    \end{split}
\end{align}    
\end{definition}

\subsection{Solving the Region Optimization via Genetic Algorithm}

\begin{wrapfigure}{H}{.33\linewidth}
   \centering
   \vspace{-10px}
     \includegraphics[trim={0mm 0mm 0mm 0mm},clip,width=\linewidth]{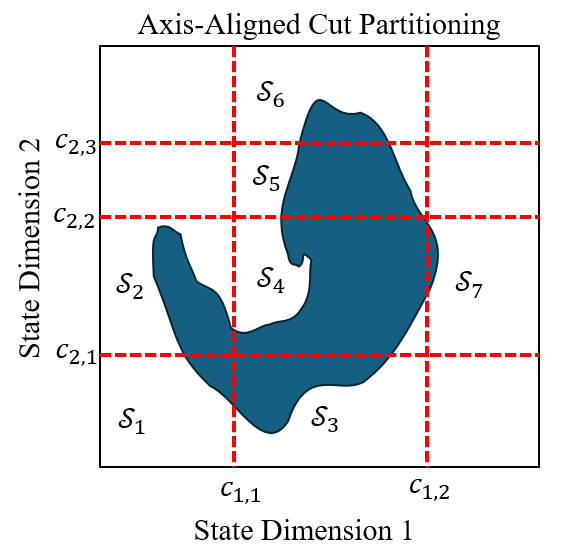}
     \caption{Example of axis-aligned cuts for 7 subregions: 2 cuts in the first state dimension and 3 cuts in the second state dimension.}
    \label{fig:state_partitioning}
 \end{wrapfigure}

The optimization problem in~\eqref{eq:optimziation_cp_region} becomes intractable if the regions $\mathcal{S}_i$ are allowed to have arbitrary shapes.
This greatly complicates the optimization and risks severe overfitting. To make the problem tractable, we restrict each partitioned region $\mathcal{S}_i$ to be an \textit{axis-aligned box}.
Such box-shaped regions are easy to parameterize, optimize, and integrate with subsequent reachability analysis, while still providing tight confidence bounds.

As an initial step, we propose to solve~\eqref{eq:optimziation_cp_region} by choosing sets of axis-aligned cuts $\mathcal{C}$ in each of $d$ state dimensions by the number of cuts $n_1,\dots,n_d$, possibly different in each dimension. Formally, we define the sets of cuts per dimension: $\{\mathcal{C}_u\}_{u=1}^{d}$ where for each dimension $u$, $\mathcal{C}_u=\{c_{u,1},\dots,c_{u,n_u}\}$ gives the individual cut locations in that dimension. Figure~\ref{fig:state_partitioning} illustrates example cut sets in 2 dimensions.

These cut sets then define disjoint, non-empty regions $\{\mathcal{S}_i\}_{i=1}^M$. The number of regions, $M$, may vary given the same number of cuts in each dimension as some of the partitioned boxes may not contain trajectory data from $D$. In Figure~\ref{fig:state_partitioning}, the top left and top right regions aree empty. To accurately determine the number of nonempty regions based on changing cut locations, we recompute the feasible regions $\{\mathcal{S}_i\}_{i=1}^M$ at each step of the optimization based on the current cuts $\{\mathcal{C}_u\}_{u=1}^{d}$ and the dataset $D$. Next, we will discuss how to handle empty regions during the optimization process. Altogether, we solve the following optimization problem.

\begin{definition}[Axis-Aligned Reachability-Informed Region Optimization]
Given a calibration dataset of trajectories $D$, the number of desired cuts in each dimension $n_1,\dots,n_d$, and a total confidence bound $\alpha$, \emph{the axis-aligned reachability-informed region optimization problem} is to select axis-aligned cuts per state dimension that:

\begin{align}\label{eq:optimziation_cuts}
    \begin{split}
    \min_{\{\mathcal{C}_u\}_{u=1}^d,\ \{\alpha_i\}_{i=1}^M}\quad 
    & \mathcal{L}\big(D, \{\mathcal{S}_i\}_{i=1}^M,\{\alpha_i\}_{i=1}^M\big) \\
    \text{s.t.} &\ \{\mathcal{S}_i\}_{i=1}^M =\textit{getNonemtyBoxes}(\{\mathcal{C}_u\}_{u=1}^d, D)\\
    &\ \mathcal{X} = \mathcal{S}_1 \cupdot \dots \cupdot \mathcal{S}_M, \quad \alpha_i \in (0, \alpha),\quad \sum_{i=1}^M \alpha_i \le \alpha\\
    &\ \Delta_i = \{\delta^1_{\mathcal{S}_i}, \dots, \delta^{N+1}_{\mathcal{S}_i}\}, i = 1, \dots, M, \text{ and }\\
    &\ \ \eta_i = \text{Quantile}\left(\Delta_i, 1-\alpha_i \right) \text{ if } x \in \mathcal{S}_i, i = 1, \dots, M.\\
    \end{split}
\end{align}    
\end{definition}

During optimization, we utilize a subroutine $\textit{getNonemptyBoxes}()$, which only returns non-empty regions with respect to dataset $D$ given all of the possible regions formed by cut sets $\{\mathcal{C}_u\}_{u=1}^d$. As illustrated in Figure~\ref{fig:state_partitioning}, there will often be empty regions unless the trajectory data is uniformly distributed throughout the space. Therefore, $M$ is the number of regions that share the total confidence $\alpha$. During optimization, the empty regions are simply ignored. However, at verification time, we want our regions to produce a tiling of the state space so that all states reached have a corresponding conformal bound. An initial approach to create a tiling is to iteratively merge empty with the adjacent, non-empty region with lowest conformal bounds until only $M$ hyper-rectangles tile the state space.

To find optimal partitions, we use a gradient-free global method, a \textit{genetic algorithm} (GA)~\cite{mirjalili2019}.
A GA is appropriate here because the objective is non-differentiable and highly nonconvex due to discrete region assignments.  In general, the approach models solutions to the optimization problem as individuals, each with their own ``genes'' (parameters) which define their evolutionary ``fitness'' -- the lower the loss, the more fit. To compute a solution, many randomized individuals are simulated simultaneously: fit candidates ``crossover'' with other fit candidates produce new individuals, encouraging convergence of good solutions, while, at the same time, individual genes can ``mutate'', encouraging random exploration of the parameter space. 

\begin{algorithm}[t]
\small
\caption{Genetic Algorithm for Region Partitioning and Confidence Allocation}
\label{alg:ga_short}
\DontPrintSemicolon
\SetKwFunction{Crossover}{Crossover}
\SetKwFunction{Mutate}{Mutate}
\KwIn{Trajectory dataset $D$, with state dimension $d$, number of region partitions (cuts) per dimension~$n_1\,\dots,n_d$, 
global confidence $\alpha$, minimum per-region confidence $\underline{\alpha}$, population size $P$, generations $G$, crossover probabilities $p_c^{\mathcal{C}}, p_c^{\alpha}$, mutation probabilities $p_m^{\mathcal{C}}, p_m^{\alpha}$, discount $\gamma$.}
\KwOut{Best partition cuts per dimension $\{\mathcal{C}_u^\star\}_{u=1}^d$, corresponding $M$ non-empty hyper rectangular regions~$\{\mathcal{S}_i^\star\}_{i=1}^M$, and confidence allocations $\{\alpha_i^\star\}_{i=1}^M$ which minimize $\mathcal{L}$ from (\ref{eq:loss}) subject to $\sum_{i=1}^M \alpha_i=\alpha$, $\alpha_i \ge 0$.}

\BlankLine
\textbf{Initialize:} population $\mathcal{P}=\{\Theta^{(p)}\}_{p=1}^{P}$ with random cuts $\{\mathcal{C}_u\}_{u=1}^d$, compute corresponding random boxes $\{\mathcal{S}_i\}_{i=1}^M$, and assign uniform confidences $\alpha_i=\alpha/M$.\;
\ForEach{$\Theta^{(p)}\in\mathcal{P}$}{Evaluate $\mathcal{L}(\Theta^{(p)})$.}
\For{$g=1$ \KwTo $G$}{
  \textbf{Elitism:} sort $\mathcal{P}$ by increasing $\mathcal{L}$ and copy top $P/2$ individuals to $\mathcal{P}_{\text{new}}$.\;\While{$|\mathcal{P}_{\text{new}}| < P$}{
    \textbf{Randomly Select parents} $\Theta^{(a)},\Theta^{(b)}$ from $\mathcal{P}_{\text{new}}$. \; 
    
    \textbf{Perform Crossover from Procedure~(\ref{alg:crossover})}:  $\tilde{\Theta} \leftarrow$ \Crossover{$\Theta^{(a)},\Theta^{(b)},p_c^{\mathcal{C}}, p_c^{\alpha}, \underline{\alpha}, \alpha,  D$}\;

     \textbf{Perform Mutation from Procedure~(\ref{alg:mutate}): $\tilde{\Theta} \leftarrow$ \Mutate{$\tilde{\Theta}, p_m^{\mathcal{C}}, p_m^{\alpha}, \underline{\alpha},\alpha, D$}} \;
\textbf{Evaluate $\mathcal{L}(\tilde{\Theta})$ and insert into $\mathcal{P}_{\text{new}}$.}\;
  }
  \textbf{Replace} $\mathcal{P}\leftarrow \mathcal{P}_{\text{new}}$.\;
}
\Return{the individual in $\mathcal{P}$ with the smallest $\mathcal{L}$}\;
\end{algorithm}

Algorithm~\ref{alg:ga_short}
outlines our GA implementation at a high level. Each individual $\Theta$ is defined as a tuple of the cut sets, $\{\mathcal{C}_u\}_{u=1}^d$, corresponding regions, $\{\mathcal{S}_i\}_{i=1}^{M}$, and per-region confidences $\{\alpha_i\}_{i=1}^M$. The fitness of each individual is defined as $\mathcal{L}(\Theta) = \mathcal{L}(D, \{\mathcal{S}_i^{\Theta}\}_{i=1}^{M},\{\alpha_i^{\Theta}\}_{i=1}^M)$ where $\mathcal{L}$ is the loss defined in (\ref{eq:loss}), applied to the individual's parameters.  
Procedures~\ref{alg:crossover}~and~\ref{alg:mutate} describe the details of the Crossover and Mutate procedures, respectively. Overall, the parameters for each individual are the boundaries of axis-aligned hyper rectangles or "cuts" as well as the per-region confidences $\alpha_i$. This partitioning and confidence allocation scheme can introduce challenges finding feasible solutions. We give several methods to help find viable candidates. 
First of all, it is possible to find a candidate region partition $\mathcal{S}_i$ for which there is insufficient data present in the region, and the resulting quantile bound $\eta
_i$ for the desired confidence $\alpha_i$ is infinite. While such high-loss solutions are often excluded by selecting the top half of individuals each generation, the repair functions help guide solutions away from this undesirable case. In particular, each time a region partition is randomly altered (when producing a new offspring or mutating), we repair its boundaries to ensure they are feasible and not redundant. Similarly, when randomly perturbing confidences, we encourage a lower bound on confidence, $\underline{\alpha}$, then normalize and scale the confidences to help prevent solutions with infinite loss. Such a lower bound can be computed such that the $1-\underline{\alpha}$ quantile value is not infinite for the number of data points used to solve the optimization problem. 

\begin{figure}[t] 
\begin{minipage}{0.49\textwidth}
\renewcommand{\algorithmcfname}{Procedure}
\setcounter{algocf}{0}
\begin{algorithm}[H]
\small
\caption{Crossover}
\label{alg:crossover}

\DontPrintSemicolon

\KwIn{Parents  $\Theta^{(a)} = (\{\mathcal{C}^{(a)}_u\}_{u=1}^{d},\{\mathcal{S}^{(a)}_i\}_{i=1}^{M^{(a)}}, \{\alpha^{(a)}_i\}_{i=1}^{M^{(a)}}),$ $\Theta^{(b)}=(\{\mathcal{C}^{(b)}_u\}_{u=1}^{d},\{\mathcal{S}^{(b)}_i\}_{i=1}^{M^{(b)}}, \{\alpha^{(b)}_i\}_{i=1}^{M^{(b)}})$; cut and confidence crossover probabilities $p_c^{\mathcal{C}}, p_c^{\alpha}$; minimum per-region confidence $\underline{\alpha}$; global confidence $\alpha$; and trajectory data $D$.}
\KwOut{Child: $\tilde{\Theta}$.}

\SetKwFunction{Crossover}{Crossover}
\SetKwProg{Fn}{Procedure}{:}{}

\Fn{\Crossover{$\Theta^{(a)},\Theta^{(b)},p_c^{\mathcal{C}}, p_c^{\alpha}, \underline{\alpha}, \alpha, D $}}{ 
    \textbf{Crossover Cuts:} \;
    Initialize child cuts to parent $(a)$'s: \hspace{0.5em}$\{\tilde{\mathcal{C}}_u\}_{u=1}^{d} \gets \{\mathcal{C}^{(a)}_u\}_{u=1}^{d}$\;
    \ForEach{$\tilde{c}_{u,i}, ~~c^{(b)}_{u,i} \in \{\tilde{\mathcal{C}}_u\}_{u=1}^{d}, ~~ \{\mathcal{C}^{(b)}_u\}_{u=1}^{d}$}{
    Crossover with parent $(b)$:\;
    \hspace{0.5em}$\tilde{c}_{u,i}\gets c^{(b)}_{u,i}$ w.p. $p_c^{\mathcal{C}}$; otherwise $\tilde{c}_{u,i}$
    }

    Repair cuts to remove duplicates:
    \hspace{0.5em}$\{\tilde{\mathcal{C}}_u\}_{u=1}^d \leftarrow \textit{repairCuts}(\{\tilde{\mathcal{C}}_u\}_{u=1}^d)$ \;
    Compute Child Regions: 
    \hspace{0.5em}$\{\tilde{\mathcal{S}_i}\}_{i=1}^{\tilde{M}} \leftarrow\textit{getNonemptyBoxes}(\{\tilde{\mathcal{C}}_u\}_{u=1}^d, ~D)$\;
    \BlankLine
    \textbf{Crossover Confidences: }\;
    \If{$\tilde{M} = M^{(a)} =M^{(b)}$}{
    Initialize child confidences to parent $(a)$'s: \hspace{0.5em}$\{\tilde{\alpha}_i\}_{i=1}^{\tilde{M}} \gets \{\alpha^{(a)}_i\}_{i=1}^{M^{(a)}}$\;
    \ForEach{$\tilde{\alpha}_{i}, ~~\alpha^{(b)}_i\in \{\tilde{\alpha}_i\}_{i=1}^{\tilde{M}}, ~~ \{\alpha^{(b)}_i\}_{i=1}^{M^{(b)}}$}{
    Crossover with parent $(b)$:\;
    \hspace{0.5em}$\tilde{\alpha}_{i}\gets \alpha^{(b)}_{i}$ w.p. $p_c^{\alpha}$; otherwise $\tilde{\alpha}_{i}$
    }

     Repair conf with $\tilde{\alpha}_i \in [\underline{\alpha}, \alpha]$ and $\sum_i \tilde{\alpha}_i=\alpha$:\;
    \hspace{0.5em}$ \{\tilde{\alpha}_i\}_{i=1}^{\tilde{M}}  \leftarrow \textit{repairConfidences}(\{\tilde{\alpha}_i\}_{i=1}^{\tilde{M}})$ \;}
     \Else{Assign uniform child confidences:\; \hspace{0.5em}$\tilde{\alpha}_i \gets \alpha/\tilde{M} ~~\forall i=1..\tilde{M}$ \;
     }
    $\tilde{\Theta} \leftarrow ( \{\tilde{\mathcal{C}}_u\}_{u=1}^d, \{\tilde{\mathcal{S}_i}\}_{i=1}^{\tilde{M}}, \{\tilde{\alpha}_i\}_{i=1}^{\tilde{M}}) $ \;
    \Return $\tilde{\Theta}$ \;
}
\end{algorithm}

\end{minipage}
\hfill
\begin{minipage}{0.49\textwidth}

\renewcommand{\algorithmcfname}{Procedure}
\begin{algorithm}[H]
\small
\caption{Mutation}

\label{alg:mutate}

\DontPrintSemicolon

\KwIn{Individual \hspace{12mm}$\Theta := ( \{\mathcal{C}_u\}_{u=1}^d, \{\mathcal{S}_i\}_{i=1}^{M}, \{\alpha_i\}_{i=1}^{M})$; cut and confidence mutation probabilities $p_m^{\mathcal{C}}, p_m^{\alpha}$; minimum per-region confidence $\underline{\alpha}$; global confidence $\alpha$; and trajectory data $D$.}
\KwOut{Mutated Individual $\tilde{\Theta}$.}

\SetKwFunction{Mutate}{Mutate}
\SetKwProg{Fn}{Procedure}{:}{}

\Fn{\Mutate{$\Theta, p_m^{\mathcal{C}}, p_m^{\alpha}, \underline{\alpha},\alpha, D $}}{
    \BlankLine 
    \textbf{Mutate Cuts:} \;
    Initialize mutated cuts: \hspace{0.5em} $\{\tilde{\mathcal{C}}_u\}_{u=1}^d \gets \{\mathcal{C}_u\}_{u=1}^d $\;
    \ForEach{$\tilde{c}_{u,i} \in \{\tilde{\mathcal{C}}_u\}_{u=1}^d $}{
    Let $c_m$ be a random cut. \;
    Mutate: \;
    \hspace{0.5em}$\tilde{c}_{u,i}\gets c_m$ w.p. $p_m^{\mathcal{C}}$; otherwise $\tilde{c}_{u,i}$
    }    
    Repair cuts to remove duplicates:  \hspace{0.5em}$\{\tilde{\mathcal{C}}_u\}_{u=1}^d \leftarrow \textit{repairCuts}(\{\tilde{\mathcal{C}}_u\}_{u=1}^d)$ \;
    Compute new regions: \hspace{0.5em} $\{\tilde{\mathcal{S}_i}\}_{i=1}^{\tilde{M}} \leftarrow \textit{getNonemptyBoxes}(\{\tilde{\mathcal{C}}_u\}_{u=1}^d, ~D)$\;
    \BlankLine
    \BlankLine 
    \textbf{Mutate Confidences:} \;
\If{$\tilde{M} = M$}{
    Initialize mutated confidences: \hspace{0.5em} $\{\tilde{\alpha}_i\}_{i=1}^{\tilde{M}} \gets \{\alpha_i\}_{i=1}^M $\;
\ForEach{$\tilde{\alpha}_{i}\in \{\tilde{\alpha}_i\}_{i=1}^{\tilde{M}}$}{
    Let $\alpha_{m} \in[\underline{\alpha},\alpha]$ be a random confidence. \;
    \BlankLine 
    Mutate:\;
    \hspace{0.5em}$\tilde{\alpha}_{i}\gets \alpha_m$ w.p. $p_m^{\alpha}$; otherwise $\tilde{\alpha}_{i}$
    } 
     Repair conf with $\tilde{\alpha}_i \in [\underline{\alpha}, \alpha]$, $\sum_i \tilde{\alpha}_i=\alpha$: \; \hspace{0.5em}
      $ \{\tilde{\alpha}_i\}_{i=1}^{\tilde{M}}  \leftarrow \textit{repairConfidences}(\{\tilde{\alpha}_i\}_{i=1}^{\tilde{M}})$ \;}
     \Else{Assign uniform child confidences:\; \hspace{0.5em}$\tilde{\alpha}_i \gets \alpha/\tilde{M} ~~\forall i=1..\tilde{M}$ \;
     }
    $\tilde{\Theta} \leftarrow ( \{\tilde{\mathcal{C}}_u\}_{u=1}^d, \{\tilde{\mathcal{S}_i}\}_{i=1}^{\tilde{M}}, \{\tilde{\alpha}_i\}_{i=1}^{\tilde{M}}) $ \;
    \Return $\tilde{\Theta}$ \;
}
\end{algorithm}
\end{minipage}
\vspace{-3mm}
\end{figure}

\section{High-Confidence Reachability Analysis}
\label{sec:reachability}
We have now obtained high-confidence bounds encoded by the perception bound function $H$ from~\eqref{eq:global_bound} for the unknown random noise $v_k$ in the system in~\eqref{eq:system_model} in regions $\mathcal{S}_1, \dots, \mathcal{S}_M$. From Theorem~\ref{thm:reachability}, in the case of $L_\infty$ norm bounds on $v_k$, we have that $v_k \in [-\eta_i, \eta_i]$ if $x_k \in S_i$. Thus, the system for which we need to compute reachable sets becomes: 
\begin{align}
\label{eq:abst_verification_model}
\begin{split}
    x_{k+1} = f(x_k, u_k); \quad 
    y_k = g(x_k) + [-\eta_i,\eta_i] ~~ \text{ if } x_k \in S_i; \quad
    u_k = h(y_k).
\end{split}
\end{align}

To compute high-confidence reachable sets for the system in~\eqref{eq:abst_verification_model}, we can use any reachability tool for non-linear hybrid systems. One such tool is the authors' tool Verisig \cite{ivanov2021b}. To encode the regional perception bounds in a hybrid system, we add transitions between the plant $f$ and controller $h$ with guards determined by the regions ($x_k\in \mathcal{S}_i$) and resets that inflate the measurement model $g(x_k)$ with the corresponding error interval $[-\eta_i, \eta_i]$.

Introducing additional transitions and adding uncertainty directly to reachable sets introduces a \textit{scalability challenge}: reachable sets that intersect with multiple regions must be considered separately. This leads to longer verification time as each ``branch'' must be verified separately. However, in verification of highly-branching systems, the reachable sets often have significant overlap (and some may even be strict subsets of others), leading to redundant and time-costly verification \cite{frehse2011spaceex}. Though it is not trivial to determine if reachable sets overlap in general, one can design heuristic measures of reachable set similarity and merge branches if they look too similar. 

Based on this intuition, we introduce a \textit{cluster-and-enclose method} for improving scalability in highly-branching hybrid system verification tasks: we merge similar (overlapping) reachable sets by first identifying similar branch clusters based on a Euclidean distance between their reachable set centroids, and then enclosing them with a new reachable set, centered around the average of the branches to be enclosed. First, we will describe our method for finding an enclosure of multiple reachable sets. Then, we will describe a clustering method to group branches based on similar reachable sets such that the additional overapproximation error induced by enclosing them is minimal.

\subsection{Enclosure Method for Taylor Model-Based Reachable Sets}

Following the TM Definition~\ref{def:TMs}, we omit the specifics of the TM operations and refer the reader to~\cite{makino03} for details. Our methods only use addition, subtraction, and scalar multiplication, which simply require applying the analogous operations to the polynomial and interval components independently. Moreover, we utilize interval evaluation of the TMs to produce bounding boxes \cite{makino03}. This involves substituting the interval domains of the TMs into the polynomials and conducting interval analysis to get a single overapproximation of the range of a given TM. We next propose a method to enclose a set of TMs into a single reference TM with an inflated remainder to make a valid enclosure.

To form a valid enclosure, we start from a reference TM $\bar{T}$ around which to enclose the $K$ TMs. To account for the maximum and minimum deviations from any of the $K$ models and the reference model, we create an inflated interval. While the true maximum and minimum deviations are difficult to describe analytically, can conduct an interval evaluation on the difference TMs to find conservative bounding boxes of the true difference sets. Taking the minimum infimum and maximum supremum of the bounding boxes ensures a valid enclosure. In particular: 
\begin{definition}[Taylor Model Union Enclosure]
Given $K$ Taylor Models $T_1,\dots, T_K$ defined on a common domain $D$ overapproximating the same function $f$ and a reference model $\overline{T}=(\overline{p}, \overline{I})$, a Taylor Model union enclosure of the $K$ TMs about $\bar{T}$ is $T_\cup = (p_\cup, I_\cup)$ such that:
$T_\cup \supseteq \{T_1,\dots,T_n\}$ and
\begin{align}\label{eq:TM_enclosure}
    \begin{split}        
        p_\cup =\overline{p}, \quad
        ~~ I_\cup = \overline{I} 
         ~+~ \left[\min_{i=1,\dots, K} \{0, \inf (T_i-\overline{T})\},~ \max_{i=1,\dots,K} \{0, \sup(T_i-\overline{T})\}\right]
    \end{split}
\end{align}
\end{definition}
In effect, the inflated interval in~\eqref{eq:TM_enclosure} accounts for the largest positive and negative deviations any of the $K$ models can differ from the reference model. To compute the $\inf$ and $\sup$ in practice, we first conduct interval analysis on the difference models $T_i-\bar{T}$, then find the $\inf$ and $\sup$ as the lower bound and upper bounds of the resulting bounding interval, respectively. We include $0$ in the minimum and maximum operators for the case where all the $K$ models are above or below the reference model.

For clarity, our definition assumes scalar-valued functions, but TMs can also be vector-valued. However, such vector-valued functions can be interpreted as vectors of scalar-valued functions, with the dimension of the polynomial and interval components increasing accordingly. 

While we don't specify the reference TM $\bar{T}$ in this definition, we choose as our reference the average of the $K$ TMs: $\overline{T} =  \frac{1}{K}\sum_{i=1}^K T_i$, where, according to TM arithmetic, $\overline{T} = (\overline{p}, \overline{I})= ( \frac{1}{K}\sum_{i=1}^K p_i,\frac{1}{K}\sum_{i=1}^K I_i) $. This choice provides a reference model that is, by definition, ``close'' to all of the models to be enclosed. A discussion of other choices of reference model can be found in Section~\ref{sec:conclusion}. Now that we have a method to compute an enclosure of $K$ TMs, we seek to find similar TMs that can be enclosed together, reducing the number of redundant reachable set branches and thereby improving verification time.

\renewcommand{\algorithmcfname}{Algorithm}
\setcounter{algocf}{1}

\begin{algorithm}[ht!]
\small
\caption{Cluster and Enclose Algorithm for Verification Branch Reduction}
\label{alg:clustering}
\DontPrintSemicolon
\SetKwFunction{Crossover}{Crossover}
\SetKwFunction{Mutate}{Mutate}
\SetKwFunction{KMEANS}{K-MEANS}
\KwIn{TM vectors for all $K$ branches  $\mathcal{T} = \{\bm{T}_{1},\dots,\bm{T}_K\}$ where $\bm{T}_i = [T_{i,1}, \dots,T_{i,d} ]$ are the TMs for each of $d$ state variables, and maximum branches desired $B$.}
\KwOut{TM vector enclosures for $b =\min (K,B)$ branches and $d$ state variables ~$\tilde{\mathcal{T}} = \{\tilde{\bm{T}}_{1},\dots,\tilde{\bm{T}}_b\}$.}
\textbf{Initialize:} Output Set $\tilde{\mathcal{T}} = \emptyset $, Reachable set midpoints $\bm{\mathcal{X}} = \emptyset$.
\BlankLine
\If{$K > B$}{
    \ForEach{$\bm{T}_i \in \mathcal{T}$}{
    Compute interval bounding Boxes $\bm{\mathcal{B}}_{i} \leftarrow \textit{intervalEval}(\bm{T}_{i})$. \;
    Let $\overline{\bm{B}_i},\underline{\bm{B}_i} \in \mathbb{R}^d$ be the upper and lower bounds of $\mathcal{B}_i$, respectively.\;
    Compute midpoints $\bm{x}_{i} = (\underline{\bm{B}_i}+\overline{\bm{B}_i})/2$. \;
    $\mathcal{X} \leftarrow \mathcal{X} \bigcup \{\bm{x}_i\}$
    }
    Get index assignments for $B$ clusters: $\mathcal{I} = \{I_1, \dots, I_B \} \leftarrow $ \KMEANS{$\mathcal{X}$}\;
    \ForEach{\text{Cluster:} $I \in \mathcal{I}$}{
    \If{$\mid {I} \mid ~ =1$  }{
        Single branch, $\tilde{\bm{T}}$, in cluster, no enclosure required.\; 
        $\tilde{\mathcal{T}} \leftarrow \tilde{\mathcal{T}} \bigcup \{ \tilde{\bm{T}}\}$\;
    }
    \Else{
        Initialize $\tilde{\bm{T}} \gets [ ~]$\;
        \For{$j=1,\dots,d$}{
            Define set to be enclosed: $\mathcal{E} = \{T_{i,j}\mid i\in I\}$\;
            Compute Union Enclosure TM from (\ref{eq:TM_enclosure}) for state $j$ over cluster $I$: $E \leftarrow \textit{computeUnionEnclosure}(\mathcal{E})$ \;
            \If{$E$ has large remainder}{
                $E \leftarrow \textit{shrinkWrap}(E)$
            }
            Append $E$ to $\tilde{\bm{T}}$\;
        }
        $\tilde{\mathcal{T}} \leftarrow \tilde{\mathcal{T}} \bigcup \{ \tilde{\bm{T}}\}$\;
    }

    }

}
\Else{
    $\tilde{\mathcal{T}} \leftarrow \mathcal{T}$\;  
}
Remove any subset branches of shrink-wrapped branches: $\tilde{\mathcal{T}} \leftarrow \textit{removeSubsets}(\tilde{\mathcal{T}})$\;
\Return{$\tilde{\mathcal{T}}$}
 
\end{algorithm}

\subsection{Branch Clustering and Consolidation}
To reduce the total number of TMs, we need to \textit{find clusters of TMs} that would add minimal approximation error after the enclosure is performed. Intuitively, such clusters would have highly overlapping reachable sets, meaning the size of the added inflated interval from (\ref{eq:TM_enclosure}) would be small, as the differences between each TM and the average TM would be small. Additionally, because the number of branches may grow exponentially. Thus, the time required to verify can be prohibitively long and we set an upper limit on the number of clusters at each step. This approach introduces a key tradeoff: the fewer total clusters permitted at each step, the shorter the verification time, but the larger the overapproximation error added by enclosures.

Given a total number of branches $K$ and a desired maximum number of branches $B$, we use a $k$-means algorithm \cite{bishop2006pattern} that clusters branches based on the Euclidean distance between the midpoints of their reachable set bounding boxes. Algorithm~\ref{alg:clustering}
provides an outline of our approach. As a first step, in this work, we use the \textit{distance between midpoints} of each branch's reachable sets as a heuristic for branch similarity. Other distance metrics will be explored in future work, as outlined in the Discussion section.
At each timestep of the hybrid system verification, if the total number of branches $K$
exceeds the maximum permitted number $B$ we compute $B$ clusters of branches to be enclosed using the total $K$ branches using $k$-means. For each resulting cluster, we compute a union enclosure over each of the $d$ state dimensions as defined in (\ref{eq:TM_enclosure}). Finally, we apply \textit{shrink-wrapping} to the enclosure TM if the remainder is too large \cite{ivanov2021b}. In general, \textit{shrink-wrapping} refactors at TM into a new one with a smaller remainder. In Verisig, we consider a remainder large if its width is both larger than some $\epsilon$ (e.g., $\epsilon > 10^{-6}$) and if its width accounts for a significant portion of the overall range (e.g., more than $10\%$ of the entire range). In this case, we refactor the TM to be fully symbolic, so the polynomial term becomes linear and contains the original range with no remainder. Finally, we opportunistically check to see if any of the resulting enclosures are strict subsets of another shrink-wrapped branch. If all state variable bounding boxes for a given ``child'' branch are subsets of those of another, shrink-wrapped ``parent'' branch, we can simply remove the child branch, as it is redundant to the parent branch.

\vspace{-3mm}
\section{Case Studies} \label{sec:evaluation}
We evaluate our proposed statistical-symbolic verification framework on two representative case studies: Mountain Car (MC)~\cite{gymMC} and a 1/10th-scale autonomous racing car based on the F1-tenth/RoboRacer platform~\cite{f10thcar}. 
These two case studies are chosen for different complexities:
\begin{itemize}
    \item \textit{Mountain Car:} an image-based task where perception error primarily depends on position.
    \item \textit{Autonomous racing car:} a LiDAR-based task with perception error dependent on multiple states (the car's pose).
\end{itemize}
In both case studies, we compare our state-based approach with a state-of-the-art time-based CP method~\cite{cleaveland2024conformal}, which optimizes weighting parameters at each time step using linear complementarity programming, as discussed in Section~\ref{sec:bg_and_relwork}.
This method uses two holdout data sets. The first is used to solve for the optimal weighting parameters, and the second is used to set conformal bounds. We compute and compare against multiple baselines by increasing the number of trajectories used to solve for the weighing parameters until solver failure. 

 For each setting, we describe the experimental setup, including the perception and control models and data collection procedure. Then, using the collected data, we compute conformal prediction bounds and integrate them into reachability analysis to obtain high-confidence ($95\%$) reachable sets using both our proposed method and the time-based baseline. Finally, we quantitatively and qualitatively compare our method to the baseline using a variety of metrics to illustrate relative conservatism and scalability. In general, we compare the sizes of the reachable sets, test coverage, and time required to verify. These metrics demonstrate that our approach provides smaller reachable sets and stronger safety guarantees compared to the baseline.
\begin{figure}[t]
  \centering
  \begin{subfigure}[b]{0.2\textwidth}
    \includegraphics[trim={5mm 0mm 0mm 0mm},clip,width=\textwidth]{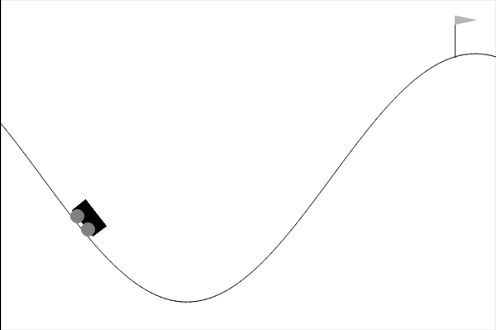}
    \caption{Canonical}
    \label{fig:mc_canonical}
  \end{subfigure}\hfill
  \begin{subfigure}[b]{0.2\textwidth}
    \includegraphics[trim={10mm 0mm 60mm 50mm},clip,width=\textwidth]{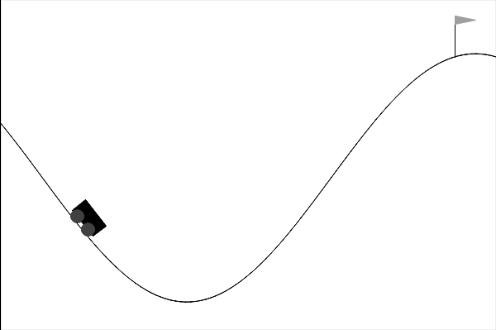}
    \caption{High-contrast}
    \label{fig:mc_hc}
  \end{subfigure}\hfill
  \begin{subfigure}[b]{0.2\textwidth}
    \includegraphics[trim={10mm 0mm 60mm 50mm},clip,width=\textwidth]{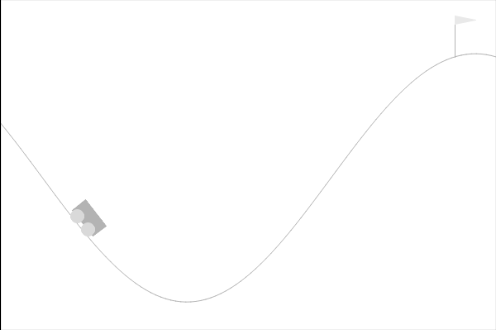}
    \caption{Low-contrast}
    \label{fig:mc_lc}
  \end{subfigure}\hfill
  \begin{subfigure}[b]{0.2\textwidth}
    \includegraphics[trim={10mm 0mm 60mm 50mm},clip,width=\textwidth]{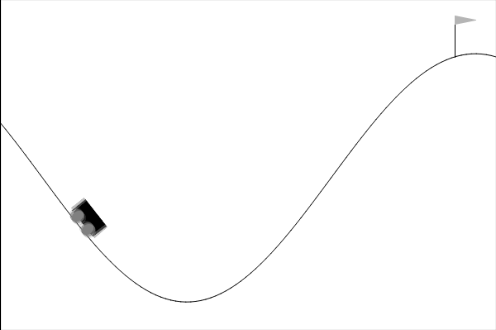}
    \caption{Blurred}
    \label{fig:mc_blurred}
  \end{subfigure}
  \caption{Examples of noise in the MountainCar environment. The noisy images were cropped to emphasize the car for visibility.}
  \label{fig:mc_examples}
  \vspace{-5mm}
\end{figure}

\subsection{Mountain Car}
Mountain Car is a standard reinforcement learning benchmark in which an underpowered car must reach the top of the right hill as shown in Figure \ref{fig:mc_canonical}. Because the car is underpowered, a successful controller must first utilize the left hill to gain momentum before reaching the goal on the right side. The dynamics for the standard system are as follows:
\begin{align}
\label{eq:mc_dynamics}
    p_{k+1} = p_k + v_k; ~~\quad
    v_{k+1} = v_k + 0.0015u_k -0.0025\cos{(3p_k)}.
\end{align}
where $p\in[-1.2, 0.6]$ is the position, $v\in [-0.07, 0.07]$ is the velocity, and $u \in [-1.0, 1.0]$ is the control thrust. The initial position is the bottom of the mountain with $p_0 \in [-0.55, -0.45]$ and at rest with $v_0 = 0$.

\subsubsection{Control and Perception Models}

For the MC case study, both the controller and the perception model are neural networks.
The controller $h$ consists of two hidden layers with 16 neurons (sigmoid activations) and a single tanh output neuron.
It was pre-trained and pre-verified in prior work~\cite{ivanov19} to reach the top of the hill with a reward of at least 90 when observing ground truth velocity and position.\footnote{The initial set verified in prior work was $p_0{\in}[-0.59, -0.45]$. In our experiments, we use the more robust range $p_0{\in}[-0.55, -0.45]$, since the controller can be unstable outside of that range.}
The perception model is a convolutional neural network (CNN) that takes grayscale $400 \times 600$ images as input, and predicts the 1D position of the car. The CNN's architecture consists of two convolutional layers with 16 channels each, two fully connected layers (with 100 neurons each and ReLU activations, and a single scaled tanh output corresponding to the MC position in $[-1.2, 0.6]$. 
To improve robustness to visual variations, it is trained on 1000 images generated from 100 positions and 9 distinct contrast levels evenly spaced from low-contrast (see Figure~\ref{fig:mc_lc}) to high-contrast (see Figure~\ref{fig:mc_hc}). The model was trained 1000 epochs with MSE loss.


\begin{figure}[t]
  \centering
  \begin{subfigure}[t]{0.36\textwidth}
    \includegraphics[trim={1mm 0mm 5mm 7mm},clip,width=\textwidth]{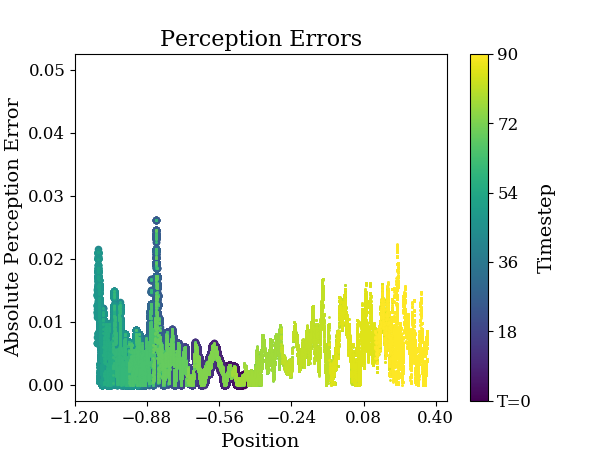}
    \caption{Perception errors for all 4,000 trajectories in $D$ showing heteroskedasticity over time and state.}
    \label{fig:mc_perception_errors}
  \end{subfigure}\hfill
  \begin{subfigure}[t]{0.36\textwidth}
    \includegraphics[trim={1mm 0mm 5mm 7mm},clip,width=\textwidth]{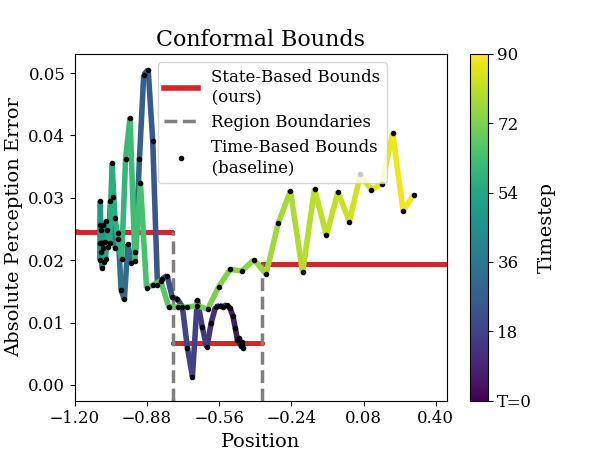}
    \caption{State-based regional conformal bounds for $M{=}3$ regions vs.\ the baseline bounds.}
    \label{fig:mc_conformal_bound_comp}
  \end{subfigure}\hfill
  \begin{subfigure}[t]{0.21\textwidth}
    \includegraphics[trim={1mm 0mm 5mm 7mm},clip,width=\textwidth]{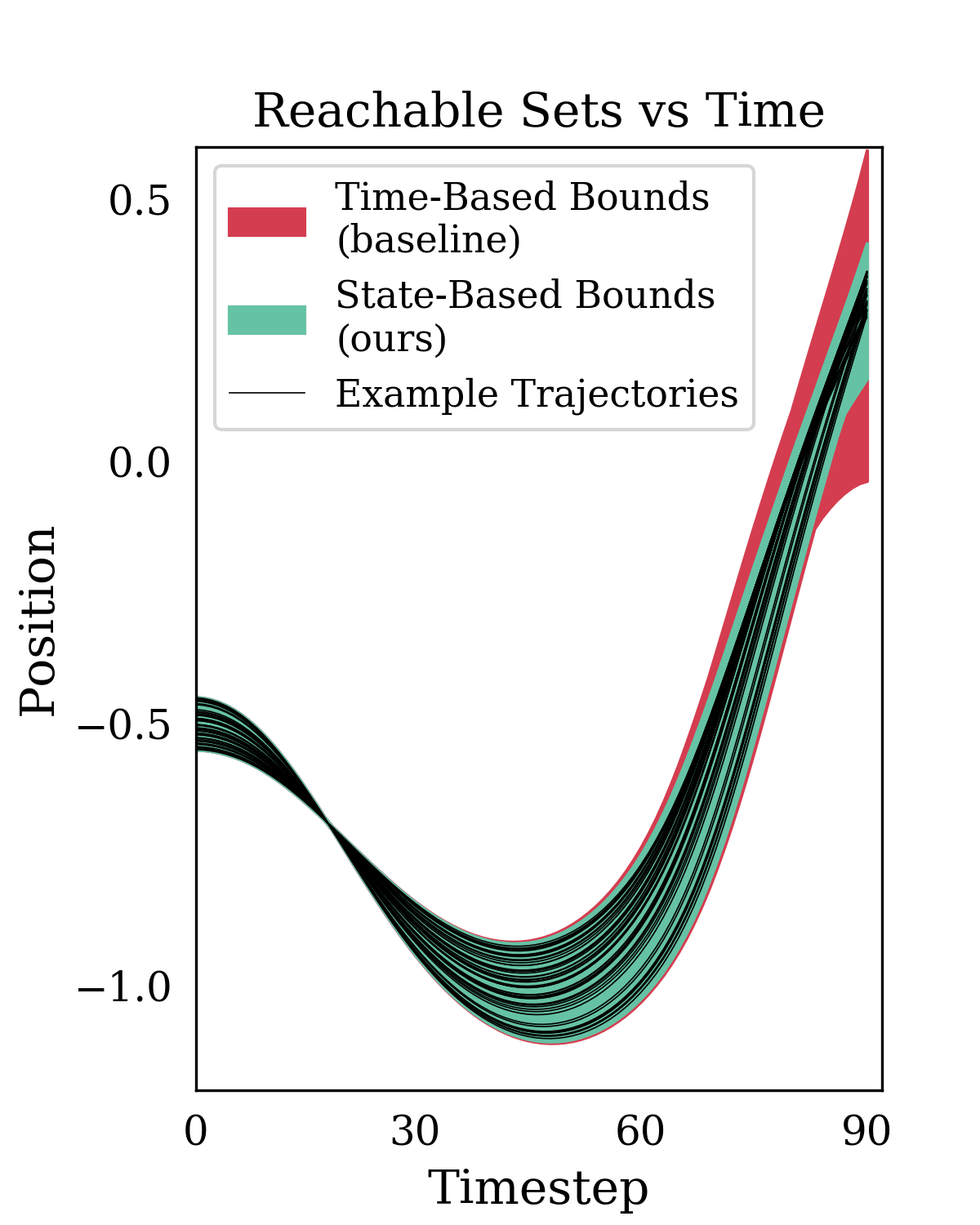}
    \caption{Reachable sets for state-based vs time-based conformal bounds.}
    \label{fig:reachtubes}
  \end{subfigure}
  \vspace{-3mm}
  \caption{Perception errors and their respective bounds under our method and a time-based baseline.}
  \label{fig:state_time_errs}
\vspace{-5mm}
\end{figure}

\subsubsection{Trajectory Data Collection}

We generate a trajectory dataset $D$ by simulating the Mountain Car environment with the perception model described previously.
To evaluate the robustness of our verification framework under realistic distribution shifts, we deploy the perception model on rollouts with blur noise applied to images (See Figure \ref{fig:mc_blurred}).
This simulates realistic perception degradation (e.g., sensor noise or poor lighting) and allows us to assess whether we can provide valid high-confidence guarantees under \emph{distribution shift}. The perception errors are computed as the absolute difference between the predicted position and the ground truth position.
The perception errors for this dataset are shown in Figure \ref{fig:mc_perception_errors} -- note the drastic heteroskedasticity over the state space exposed by the added blur noise.


Altogether, we collect 4000 trajectories from MC. The initial states are uniformly sampled from [-0.55, -0.45] and are simulated for 90 steps.
We split each trajectory dataset evenly into a 2,000-trajectory calibration set $D_{\text{cal}}$ 
and a 2,000-trajectory test set $D_{\text{test}}$. 
The calibration set $D_{\text{cal}}$ is later further divided into 
$D_{\text{reg}}$ for optimizing region partitions via~\eqref{eq:optimziation_cuts} 
and $D_{\text{conf}}$ for computing regional conformal bounds after splitting. 
The test set $D_{\text{test}}$ is reserved for evaluating the conservativeness 
of the resulting probabilistic guarantees.

\subsubsection{Conformal Bounds Computation}
We first optimize the state space partitions and subsequently calculate the corresponding regional conformal bounds with a target trajectory-wide coverage level of 95\% ($\alpha=0.05$).
To solve for the regions via~\eqref{eq:optimziation_cuts}, we randomly select 500 trajectories from $D_{\text{cal}}$ to form $D_{reg}$.
These trajectories are exclusively for optimizing regions. The remaining 1,500 trajectories are allocated to $D_{conf}$ and are used to compute the regional conformal bounds. We use Algorithm~\ref{alg:ga_short} for 1-6 cuts in the position dimension, producing results for 2 to 7 regions. In all instances of the Genetic Algorithm, we use $G{=}100$ generations, population size $P=1,000$, mutation probability $p_m^{\mathcal{C}}=p_m^{\alpha}=0.3$, crossover probability $p_c^{\mathcal{C}}=p_c^{\alpha}{=}0.5$, and discount factor $\gamma=0.9$.

For the time-based baseline, we randomly choose a holdout set of trajectories from $D_{\text{cal}}$ to solve for the weighting parameters, and use the remaining to set the time-based conformal bounds. We vary the size of $D_{\text{cal}}$ between 50, 100, 150, and 200 trajectories for parameter optimization -- no improvements were observed for larger numbers, at substantial computational cost. Figure~\ref{fig:mc_conformal_bound_comp} shows an example of the region-based conformal bounds for $M{=}3$ regions and the time-based bounds on an example trajectory. 

\begin{table}[t]
\centering
\small
\renewcommand{\arraystretch}{1.0}
\begin{tabularx}{\linewidth}{>{\centering\arraybackslash}X l >{\centering\arraybackslash}X >{\centering\arraybackslash}X >{\centering\arraybackslash}X >{\centering\arraybackslash}X}
\hline
\textbf{Cuts} & \textbf{Solution Type} & \textbf{RSS $\downarrow$} & \textbf{Max RSS $\downarrow$} & \textbf{Time [hr]$\downarrow$} & \textbf{Coverage$\uparrow$} \\
\hline
\multirow{2}{*}{6} & Fixed Confidence; $M$=7  & \textbf{4.360} & 0.115 & 256.0 & 95.2\% \\
                   & Dynamic Confidence; $M$=7 & 4.362 & \textbf{0.114} & 282.7 & 94.8\% \\
\multirow{2}{*}{5} & Fixed Confidence; $M$=6  & 4.933 & 0.154 & 159.3 & 94.7\% \\
                   & Dynamic Confidence; $M$=6 & 4.841 & 0.148 & 163.0 & 94.7\% \\
\multirow{2}{*}{4} & Fixed Confidence; $M$=5  & 5.141 & 0.163 & 131.0 & 94.7\% \\
                   & Dynamic Confidence; $M$=5 & 5.110 & 0.165 & 128.0 & 94.7\% \\
\multirow{2}{*}{3} & Fixed Confidence; $M$=4  & 5.266 & 0.168 & 107.3 & 95.4\% \\
                   & Dynamic Confidence; $M$=4 & 5.254 & 0.166 & 110.6 & 95.2\% \\
\multirow{2}{*}{2} & Fixed Confidence; $M$=3  & 6.072 & 0.203 & 100.1 & 95.3\% \\
                   & Dynamic Confidence; $M$=3 & 6.016 & 0.202 & 99.5 & \textbf{96.1}\% \\
\multirow{2}{*}{1} & Fixed Confidence; $M$=2  & 9.845 & 0.456 & 116.4 & 95.1\% \\
                   & Dynamic Confidence; $M$=2 & 9.496 & 0.434 & 96.4 & 95.2\% \\
\hline
\multirow{4}{*}{Baseline} & Time-based CP (200 traj.)~\cite{cleaveland2024conformal} & 5.508 & 0.177 & \textbf{48.4} & 94.2\% \\
                           & Time-based CP (150 traj.)~\cite{cleaveland2024conformal} & 6.219 & 0.240 & 48.9  & 95.2\% \\
                           & Time-based CP (100 traj.)~\cite{cleaveland2024conformal} & 6.611 & 0.225 & 49.1 & 95.3\% \\
                           & Time-based CP (50 traj.)~\cite{cleaveland2024conformal}  & 5.698 & 0.210 & 48.5 & 94.8\% \\
\hline
\end{tabularx}
\caption{MC case study: Comparison of fixed/dynamic confidence and time-based methods in terms of reachable set size (RSS), the max RSS at any individual timestep, verification time, and test coverage. $M$ is the number of subregions partitioned by the cuts.}
\label{tab:MC_verification_results}
\vspace{-7mm}
\end{table}

\subsubsection{Reachable Set Evaluation}
To evaluate our solution to Problem~\ref{prob:reachability}, we compute reachable sets via reachability analysis in Verisig over the 90 steps. 
To compute reachable sets with our state-based conformal bounds, we first get the regional conformal bounds following the approach in Section~\ref{sec:state_regions}. 
For the time-based baseline method, a different perception error bound is applied at each time step, as per the bounds shown in Figure~\ref{fig:mc_conformal_bound_comp}. 
For visual evaluation, we show the the full reachable sets produced by our method ($M{=}7$) as compared to the time-based baseline in Figure~\ref{fig:reachtubes}.

The full comparison between the two methods is provided in Table~\ref{tab:MC_verification_results}, where we vary the number of cuts used in our approach, as well as the number of trajectories used in the time-based baseline. We compare the methods in terms of reachable set size (RSS) by totalling the widths of the reachable set intervals at each timestep. We also report the Maximum RSS over any timestep as well as the test coverage achieved by our method against the baseline.\footnote{For each experiment, the verification for the initial set $\mathcal{X}_0$ was carried out in parallel with 200 sub-intervals of size $10^{-4}$.} 
Test coverage is measured using 2,000 trajectories in $D_\text{test}$ to determine the proportion of trajectories contained within the noise error bounds, reflecting their conservativeness. 
Ideal bounds should achieve high test coverage and small maximum RSS while maintaining efficient verification time.  

\subsubsection{Results and Discussion}
In this case study, our approach results in significantly smaller reachable set sizes than the time based approach while still maintaining high test coverage. Specifically, our best method in terms of RSS (Fixed Confidence; $M=7$ regions) attains a set size of $4.360$, significantly smaller than the best time-based method which achieved an RSS of $5.508$. Importantly, our method attains the smaller reachable set size with higher test coverage ($95.2\%$ as compared to $94.2\%$ for the time-based method). This demonstrates our method can be further improved if tighter conformal bounds are obtain (as discussed at the end of this section).

In general, our method outperformed the baselines in all settings with at least $M=4$ regions in terms of both set size metrics. However, the time required to verify generally increases with the number of regions considered, as branches are induced due to additional region boundaries. The time-based approach does not suffer from branching as the conformal error bound is applied at each timestep and does not depend on the reachable set itself, which intersects multiple regions more frequently as its size increases in the state-based case. The dynamic confidence provides improvement in RSS and max RSS against the fixed confidence solutions in all but one setting, with a negligible impact on verification time or test coverage overall.
While we did apply our \textit{cluster-and-enclose} method to this case study with a maximum number of branches $B=300$, the number of branches per timestep never exceeded that number. The influence of this scalability improvement is more obvious in the following higher-dimensional case study.

\vspace{-3mm}
\subsection{Autonomous Racing Car}
The second case study is on a 1/10-scale autonomous racing car that uses LiDAR observations with neural perception to navigate a turn in a racing track. The car dynamics are borrowed from \cite{ivanov20a} and are assumed to follow a bicycle model (with no slipping/friction) as follows:
\begin{align}
  \label{eq:f1_dynamics}
    \dot{x} = v cos(\theta); ~~\quad 
    \dot{y} = v sin(\theta); ~~\quad
    \dot{v} = - c_av + c_ac_m(u - c_h); ~~\quad
    \dot{\theta} = \frac{V}{l_f + l_r}tan(\delta).
\end{align}
where $v$ is linear velocity, $\theta$ is orientation, $x$ and $y$ are the two-dimensional position; $u = 16$ is the constant throttle input, and $\delta$ is the heading input (ranging between -15 and 15 degrees); $c_a$ is an acceleration constant, $c_m$ is a car motor constant, $c_h$ is a hysteresis constant, and $l_f$ and $l_r$ are the distances from the car's center of mass to the front and rear, respectively. Parameter values were identified by \cite{ivanov20a} as: ${c_a = 1.633}, {c_m = 0.2}, {c_h = 4}, {l_f = 0.225m}, {l_r = 0.225m}$. Note that the car model is in continuous time; however, the controller is sampled at discrete time steps (at 10Hz), the system effectively operates in discrete time.


%
\subsubsection{Perception Pipeline}
The perception model for the LiDAR-based autonomous racing car consists of a CNN followed by a particle filter \cite{arulampalam2002tutorial}. Given a 21-dimensional LiDAR input scan, the CNN is trained to predict the signed error in the relative heading and distance of the car to a reference trajectory. See Figure~\ref{fig:f110_error_diagram} for an illustration of the predicted values. The CNN uses three 1D convolutional + max pooling layers with the following number of channels: $[64,128,128]$. Finally, the network has two fully-connected layers with 50 neurons each and two output neurons. All internal activations are ReLU, and the output neurons have Tanh activations

Inferring the relative position of the car against the reference trajectory is a difficult task for a neural network alone as small changes in the position of the car can result in dramatic changes in the input lidar scans \cite{sarker2024comprehensive}. To address this challenge, we additionally apply a particle filter to our observations~\cite{thrun_probabilistic_2002}. For the particle filter, we simulate 2,000 particles with process (dynamics) noise of $0.005$ on each particle, and measurement standard deviation $7.5\times10^{-5}$ for selecting particles. The particles are instantiated uniformly over the entire initial position set, a $0.2m\times 1m$ box at the center of the first hallway (see Figure \ref{fig:f11_perception_errors}). We define the entire perception pipeline (including the particle filter) with $g: \mathbb{R}^{21}\mapsto \mathbb{R}^2$ which takes an input LiDAR scan $y\in\mathbb{R}^{21}$ and outputs estimates of heading and distance errors $(\hat{e}_{\theta},\hat{e}_d)$ against the reference trajectory. The model was trained with trained MSE loss on 2,998,195 pairs and of lidars scans and associated reference errors, generated by discretizing the 3D position and heading throughout the hallway environment.

\subsubsection{Controller.}
The controller $h:\mathbb{R}^2 \mapsto\mathbb{R}$
is a proportional bang-bang controller which takes in the heading and distance errors and outputs the maximum positive or negative steering action based on the sign of the total error term: $\delta = h(\hat{e}_\theta, \hat{e}_d) = \delta_{\mathrm{max}}\cdot\mathrm{sign}(\kappa_p (\hat{e}_\theta+ \hat{e}_d))$.  In this case study, the maximum steering angle is $\delta_{\mathrm{max}}{=}15 $ degrees, and the proportionality constant is $\kappa_p{=}1$. This controller model induces discretely different control actions based on the sign of the error term and thus branches significantly during verification when the reachable set for the error term contains 0.

\begin{figure}[t]
  \centering
  \begin{subfigure}[b]{0.26\textwidth}
    \includegraphics[trim={0mm 10mm 0mm 0mm},clip,width=\linewidth]{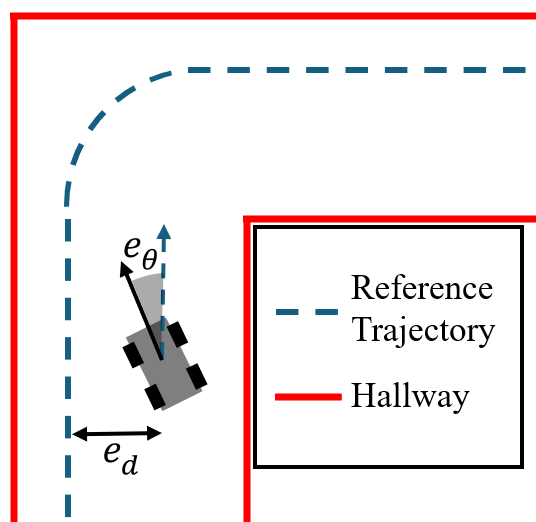}
    \caption{Illustration of heading and position errors against reference trajectory.}
    \label{fig:f110_error_diagram}
  \end{subfigure}\hfill
  \begin{subfigure}[b]{0.32\textwidth}
    \includegraphics[trim={1mm 0mm 5mm 7mm},clip,width=\textwidth]{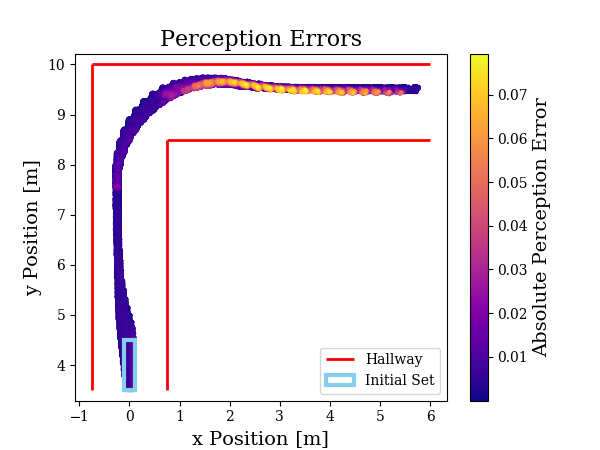}
    \caption{Perception errors for all 4,000 trajectories in $D$ showing heteroskedasticity over time and state.}
    \label{fig:f11_perception_errors}
  \end{subfigure}\hfill
  \begin{subfigure}[b]{0.32\textwidth}
    \includegraphics[trim={1mm 0mm 5mm 7mm},clip,width=\textwidth]{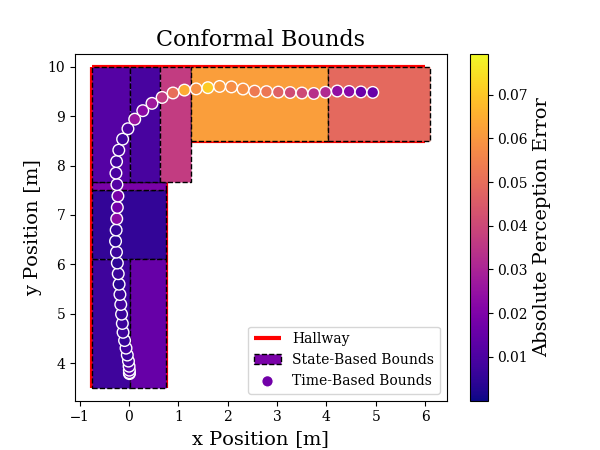}
    \caption{State-based dynamic regional conformal bounds for 4$x$ cuts and 3$y$ cuts with $M{=}9$ regions vs.\ the baseline bounds.}
    \label{fig:f110_conformal_bound_comp}
  \end{subfigure}
  \caption{Perception errors and their respective bounds under our method and a time-based baseline.}
  \label{fig:f110_state_time_errors}
  \vspace{-5mm}
\end{figure}


\subsubsection{Trajectory Data Collection}
We generate a dataset $D$ of 4,000 trajectories with initial positions in a $0.2m\times 1m$ box at the center of the first hallway (see Figure \ref{fig:f11_perception_errors}) and run for 50 timesteps with a fixed throttle and the perception and control pipeline defined above. The perception errors are computed as the absolute difference between the total reference error $\hat{e} = \hat{e}_\theta+ \hat{e}_d$ and the ground truth total reference error. Figure~\ref{fig:f11_perception_errors} shows the heteroskedasticity of the perception errors with respect to 2D position. The majority of the error is concentrated in regions near the corner, where perception is more difficult.

\subsubsection{Conformal Bound Computation}
To compute the regions and associated conformal bounds, we again split the trajectory dataset evenly into a 2,000-trajectory calibration set $D_{\text{cal}}$ 
and a 2,000-trajectory test set $D_{\text{test}}$. 
The calibration set $D_{\text{cal}}$ is further divided evenly into a 1,000 trajectory 
$D_{\text{reg}}$ set for optimizing region partitions via~\eqref{eq:optimziation_cuts} 
and a 1,000 trajectory set $D_{\text{conf}}$ for computing regional conformal bounds of the optimized regions. The test set $D_{\text{test}}$ is again reserved for computing test coverage. 

While the state space has four dimensions in total, we choose to optimize regions in two of the dimensions: $x$ and $y$. There are many ways to set cuts in each of the two dimensions. For example, with two total cuts, there are three settings: two in $x$, none in $y$; one in each $x$ and $y$; and none in $x$ and two in $y$. In total, we ran all possible settings for 2-7 total cuts. Then, for each number of total cuts, we chose the solution with the lowest loss. Finally, for these best solutions, we compute the conformal bounds per region using the remaining 1,000 trajectories in the $D_{conf}$ dataset. This process is carried out for both dynamic and fixed confidences. Throughout, we solve~\eqref{eq:optimziation_cuts} using Algorithm~\ref{alg:ga_short} with $G{=}100$ generations, population size $P{=}1,000$, mutation probability $p_m^{\mathcal{C}}{=}p_m^{\alpha}{=}0.3$, crossover probability $p_c^{\mathcal{C}}{=}p_c^{\alpha}{=}0.5$, and discount factor $\gamma{=}0.925$.

Baseline time-based bounds are computed for four different data splits for $D_{\text{cal}}$. We compute the conformal reweighting parameters using 100, 200, 300, and 500 trajectories. Then, set bounds using the remaining data in $D_{\text{cal}}$. The solver for the 500 trajectory setting did not finish due to a memory error. Example bounds for our region-based approach and the time-based bounds are shown in Figure~\ref{fig:f110_conformal_bound_comp}.

\begin{table}[t]
\centering
\small
\renewcommand{\arraystretch}{1.0}
\begin{tabularx}{\linewidth}{>{\centering\arraybackslash}X l >{\centering\arraybackslash}X >{\centering\arraybackslash}X >{\centering\arraybackslash}X >{\centering\arraybackslash}X}
\hline
\textbf{Cuts} & \textbf{Solution Type} & \textbf{RSS [m$^2$] $\downarrow$} & \textbf{Safe Dist. [cm]$\uparrow$} & \textbf{Time [hr]$\downarrow$} & \textbf{Coverage$\uparrow$} \\
\hline
\multirow{2}{*}{7} & Fixed Confidence 4x3y; $M$=9   & 2.498 & 20.87 & \textbf{152.2} & \textbf{97.7\%} \\
                   & Dynamic Confidence 4x3y; $M$=9 & \textbf{2.456} & 20.95 & 212.5 & 97.5\% \\
\multirow{2}{*}{6} & Fixed Confidence 3x3y; $M$=8   & 2.497 & 20.85 & 213.1  & 97.6\% \\
                   & Dynamic Confidence 3x3y; $M$=8 & 2.469 & \textbf{21.12} & 211.0 & 97.4\% \\
\multirow{2}{*}{5} & Fixed Confidence 4x1y; $M$=8  & 2.496 & 20.89 & 233.3 & 97.1\% \\
                   & Dynamic Confidence 3x2y; $M$=7 & 2.483 & 21.01 & 214.0 & 97.5\% \\
\multirow{2}{*}{4} & Fixed Confidence 3x1y; $M$=6   & 2.526 & 20.12 & 213.2 & 96.7\% \\
                   & Dynamic Confidence 3x1y; $M$=6 & 2.485 & 20.15 & 208.4 & 97.5\% \\
\multirow{2}{*}{3} & Fixed Confidence 2x1y; $M$=5   & 2.537 & 20.12 & 214.0 & 96.6\% \\
                   & Dynamic Confidence 2x1y; $M$=5 & 2.498 & 20.05 & 211.6 & 97.3\% \\
\multirow{2}{*}{2} & Fixed Confidence 1x1y; $M$=4   & 2.648 & 16.03 & 231.8 & 97.3\% \\
                   & Dynamic Confidence 1x1y; $M$=4 & 2.596 & 16.55 & 225.7 & 97.1\% \\
\hline
\multirow{4}{*}{Baseline} & Time-based CP (500 traj.)~\cite{cleaveland2024conformal} & DNF   & DNF   & DNF   & DNF   \\
                           & Time-based CP (300 traj.)~\cite{cleaveland2024conformal} & 2.479 & 19.99 & 198.0 & 94.7\% \\
                           & Time-based CP (200 traj.)~\cite{cleaveland2024conformal} & 2.469 & 19.82 & 195.5 & 94.6\% \\
                           & Time-based CP (100 traj.)~\cite{cleaveland2024conformal} & 2.626 & 19.29 & 252.5 & 95.6\% \\
\hline

\end{tabularx}
\caption{F1 tenth case study: Comparison of fixed/dynamic confidence state-based solutions and time-based baseline in terms of reachable set size (RSS), safety distance, verification time, and test coverage. $M$ is the number of subregions partitioned by the cuts (e.g., cuts = 5 yields $M$=8 for fixed confidence, $M$=7 for dynamic confidence).}
\label{tab:F1_verification_results}
\vspace{-7mm}
\end{table}

\subsubsection{Reachable Set Evaluation}
\begin{wraptable}[10]{r}{.45\linewidth}
\vspace{-1mm}
\centering
\small
\renewcommand{\arraystretch}{0.9}
\begin{tabularx}{\linewidth}{>{\centering\arraybackslash}X >{\centering\arraybackslash}X >{\centering\arraybackslash}X >{\centering\arraybackslash}X}
\hline
 \textbf{Max $B$} & \textbf{RSS [m$^2$] $\downarrow$} & \textbf{Safe Dist. [cm]$\uparrow$} & \textbf{Time [hr]$\downarrow$} \\
\hline
 25 & 3.885 & 8.99 & \textbf{49.3} \\
 50 & 3.063 & 19.92 & 73.7 \\
100 & 2.614 & 20.95 & 109.4 \\
200 & 2.478 & 20.95 & 164.6\\
300 & 2.456 & 20.95 & 212.5 \\
500 & \textbf{2.446} & \textbf{20.95} & 335.8 \\
\hline
\end{tabularx}
\caption{\textit{Cluster-and-enclose}: Ablation on max branches per verification step (autonomous racing; Dynamic Confidence 4x3y, $M{=}9$).}
\label{tab:clustering_ablation}
\end{wraptable}
To compute high-confidence reachable sets as a solution to Problem~\ref{prob:reachability}, we compute reachable sets in Verisig over 50 time steps, using M discrete jumps between the dynamics and the controller, and inflating the perception error with the per-region conformal bounds found above. In all results in Table~\ref{tab:F1_verification_results}, we set the maximum number of allowable branches $B{=}300$.   

To evaluate the size and conservatism of these reachable sets, we compute two verification size metrics: RSS and verified safe distance. RSS is defined similarly to the first case study -- as is common in reachability tools, we use boxes to approximate the size of each reachable set and calculate their total area using a sweep-line algorithm~\cite{yildiz2012klee}.
Verified safe distance provides the closest distance the car comes to any of the walls during verification. Together, these two metrics provide a broader picture of reachable set size: RSS measures the overall size, and thus conservatism, of the reachable sets, and the verified safe distances measure the amount of conservatism in safety-critical areas (near the wall). When coupled with test coverage, these metrics provide a more holistic measure of tightness. Overall, the best solutions would have a small RSS, a large verified safe distance, and high test coverage.


\subsubsection{Results and Discussion}

Table~\ref{tab:F1_verification_results} summarizes our results. Overall, our proposed region-based method can produce a smaller RSS and reliably better safe distances with higher test coverages than the baseline approach with similar verification times. In particular, our best approach produces an RSS of $2.456~\text{m}^2$, while the best time-based solution achieves RSS $2.469~\text{m}^2$. Moreover, for more than 2 cuts, our method consistently outperforms the baseline in terms of verified safe distance, with all solutions verified ${>}20\text{cm}$ from the wall, a distance not achieved by time-based baselines. In all our solutions, we attain higher test coverage than the baselines, with no individual solution below $96\%$. While this is indeed higher than the target coverage of $95\%$ (likely due to the conservatism from the union bound over regions), we still achieve better reachable set sizes for an even higher coverage rate, indicating further room for improvement.

In each case, dynamic confidence improves the RSS against the fixed confidence solution, and in all but one case, also improves the verified safe distance. Intuitively, this makes sense as there are more free parameters for the optimization algorithm to search over, and thus it is easier to find a good local optimum. Importantly, all dynamic confidences also maintain (and actually exceed) the user-defined confidence level of $95\%$, indicating valid maintenance of the confidence constraints during optimization. The number of cuts and regions used above 2 cuts produce reasonably consistent results, but there is no direct correlation between the number of cuts used and the size metrics. However, this is expected as the genetic algorithm is effectively a stochastic search: as the number of dimensions, cuts, and/or regional confidences increase, more searching is required to find a good optimum with the same search parameters.

In terms of verification time, our methods are very similar to the baseline methods on average. Unlike in the MC case study, the vast majority of branching is due to the discrete bang-bang controller, rather than intersections with the region boundaries. We applied our \textit{cluster-and-enclose} procedure in the verification of all results in Table~\ref{tab:F1_verification_results} with 300 maximum branches per step. To illustrate the impact of the \textit{cluster-and-enclose} method (Algorithm~\ref{alg:clustering}) for branch reduction, we computed the metrics over a variety of values of $B$, the maximum permitted branches per step.  In Table~\ref{tab:clustering_ablation}, we show how RSS, verified safe distance, and time vary as a function of $B$. Table~\ref{tab:clustering_ablation} shows how decreasing the number of maximum branches can greatly improve verification time, but at the cost of additional over-approximation error due to forced enclosures over dissimilar branches. Notably, increasing the number of allowed branches $B$ has a continuous improvement on RSS (the total reachable set area), but has no improvement on the verified safe distance above the threshold of $B=50$. This is because the point in the trajectory where the car is closest to the wall is just after the corner (a little more than halfway through the trajectory). At that point, the number of branches has not exceeded the threshold of $50$, so there is no added overapproximation error. Instead, the overapproximation error from enclosures increases as the branches continue to grow exponentially later into the trajectory.

\vspace{-3mm}
\section{Discussion and Conclusion}\label{sec:conclusion}

Altogether, this work presented a \textit{neuro-symbolic verification framework} that integrates \textit{state-dependent conformal perception bounds} into formal reachability analysis. Unlike time-based conformal prediction methods, our region-based approach captures the heteroskedasticity of perception error over the state space, producing tighter and less conservative safety guarantees. 
We formalized the problem of reachability-informed region optimization, introduced a genetic algorithm for efficient region partitioning and confidence allocation, and proposed a branch-merging procedure to enhance verifiability of highly-branching hybrid systems. Through case studies on both image-based Mountain Car and LiDAR-based 1/10th scale autonomous racing, we demonstrated that the proposed method achieves smaller reachable set sizes as compared to state-of-the-art time-series conformal prediction for the same user-defined confidence. 

On the \textit{statistical side}, the proposed region-based conformal bounds exploit the heteroskedasticity in perception errors across the state space rather than time. By allowing dynamic confidences per region, our solutions are locally adaptive, allowing better exploitation of error heteroskedasticity than fixed confidences alone. While our methods generally result in higher test coverages than the baseline time-based approach (likely the result of applying the conservative union bound over regions), our resulting reachable sets are smaller. In this way, our methods are capable of producing notably tighter reachable sets with strong statistical coverage. An interesting future direction would be to explore the use of the linear complementarity programming~\cite{cleaveland2024conformal} to further optimize the conformal bounds over the regions, in a similar fashion to the original time-based approach. Such an optimization may further increase the tightness of our methods by avoiding the conservatism of the union bound.

On the \textit{verification side}, our \textit{cluster-and-enclose} branch reduction algorithm improves the verifiability of highly-branching hybrid systems. Such an improvement is essential for systems with many discrete modes, particularly those subject to uncertainty, like those that use noise bounds as in the presented case studies. Our method induces a key tradeoff between the number of maximum branches permitted and the amount of overapproximation error added by enclosing clusters. It also introduces a variety of future investigation directions. In particular, we note that other metrics could provide better clustering distances. One particularly promising metric would be to cluster based on the size of added overapproximation error in~\eqref{eq:TM_enclosure}. While this would likely reduce the added overapproximation error, an efficient implementation would be crucial to mitigate the additional time cost of multiple TM evaluations. Also, we assume a fixed number of maximum branches to be verified at each step. Of course, the maximum number of branches to be verified could be dynamic and instead determined by the maximum amount of overapproximation error to be induced by merging branches at each step. Finally, we chose the reference TM of our enclosures to be the average, but other choices also merit further study. For example, the reference TM could be selected to specifically minimize the size of the enclosure remainder in~\eqref{eq:TM_enclosure}. We leave the investigation of these directions for future work.


\bibliographystyle{ACM-Reference-Format}
\bibliography{all_reference.bib}

\noindent

\end{document}